%% file: markovian-persuasion.tex
\documentclass[11pt,runningheads]{llncs}
\usepackage{graphicx}
\usepackage[sectionbib]{natbib}

\usepackage[utf8]{inputenc} 
\usepackage[T1]{fontenc}    
\usepackage{hyperref}       
\usepackage{url}            
\usepackage{booktabs}       
\usepackage{amsfonts}       
\usepackage{nicefrac}       
\usepackage{microtype}      
\usepackage[dvipsnames, svgnames]{xcolor}         
\usepackage{amsmath, amssymb} 
\usepackage[all,warning]{onlyamsmath}
\usepackage{graphics, graphicx} 
\usepackage[margin=1in]{geometry} 
\usepackage{mathtools} 
\usepackage{units} 
\usepackage{microtype} 
\usepackage{enumitem} 
\usepackage{xspace}
\usepackage{subfigure}

\spnewtheorem{assumption}[theorem]{Assumption}{\bfseries}{\itshape}

\usepackage{multicol}

\usepackage{fancyhdr} 
\usepackage{url} 
\usepackage[normalem]{ulem}

\usepackage{tikz} 
\usetikzlibrary{shapes,decorations}
\usetikzlibrary{arrows}

\tikzstyle{arrow} = [very thick,->,>=stealth]
\usepackage{pgfplots} 
\pgfplotsset{compat=newest} 

\makeatletter
\renewcommand{\@Opargbegintheorem}[4]{%
  #4\trivlist\item[\hskip\labelsep{#3#2\@thmcounterend}]}
\makeatother
\usepackage{apxproof}

\usepackage{comment}

\DeclareMathOperator*{\argmax}{arg\!\max}

\newcommand{\defeq}{\ensuremath{\coloneqq}}
\newcommand{\prob}{\ensuremath{\mathbf{P}}}
\newcommand{\expec}{\ensuremath{\mathbf{E}}}
\newcommand{\one}{\ensuremath{\mathsf{1}}}

\newcommand{\ind}{\ensuremath{\mathbf{I}}}
\newcommand{\reals}{\ensuremath{\mathbb{R}}}
\newcommand{\integers}{\ensuremath{\mathbb{Z}}}

\newcommand{\bigO}{\ensuremath{\mathcal{O}}}

\newcommand{\mpp}{\ensuremath{\mathsf{MPP}}}
\newcommand{\opt}{\ensuremath{\mathsf{OPT}}}
\newcommand{\pers}{\ensuremath{\mathsf{Pers}}}
\newcommand{\steady}{\ensuremath{\mathsf{Inv}}}
\newcommand{\conv}{\ensuremath{\mathsf{Conv}}}
\newcommand{\full}{\ensuremath{\mathsf{full}}\xspace}
\newcommand{\no}{\ensuremath{\mathsf{no}}\xspace}
\newcommand{\pair}{\ensuremath{\mathcal{X}}}
\newcommand{\rp}{\ensuremath{\mathsf{RP}}}
\newcommand{\lp}{\ensuremath{\mathsf{LP}}}

\usepackage{hyperref} 
\hypersetup{
    colorlinks,
    linkcolor={red!50!black},
    citecolor={blue!50!black},
    urlcolor={blue!80!black}
}

\begin{document}
\title{Markov Persuasion Processes with Endogenous Agent Beliefs}

%
\author{Krishnamurthy Iyer\inst{1}
\and
  Haifeng Xu\inst{2}
  \and
  You Zu\inst{1}
  }
%

%
\institute{Industrial and Systems Engineering, University of Minnesota, 
\email{\{kriyer,zu000002\}@umn.edu}
\and
Department of Computer Science, University of Chicago, 
\email{haifengxu@uchicago.edu}}
\maketitle              
\input{abstract}

\input{introduction}

\input{lit-survey}
\input{model}
\input{benchmarks}

\input{robustness}

\input{conclusion}
%
%
%

%
%
%

\bibliographystyle{plainnat}
\bibliography{references}

\newpage
\appendix 
\input{appendix}
\end{document}

%% file: abstract.tex
\begin{abstract}
  We consider a dynamic Bayesian persuasion setting where a single
  long-lived sender persuades a stream of ``short-lived'' agents (receivers) by sharing
  information about a payoff-relevant state. The state transitions are
  Markovian conditional on the receivers' actions, and the sender
  seeks to maximize the long-run average reward by committing to a
  (possibly history-dependent) signaling mechanism. Such problems are
  common in platform markets, where the platform seeks to achieve desirable long-term revenue and welfare
  outcomes by influencing the actions of users. While most previous studies of Markov persuasion consider exogenous agent beliefs that are independent of the chain, we study a more natural variant with \emph{endogenous} agent beliefs that depend on the chain's realized history. A key challenge to analyzing such settings is to model the agents' partial knowledge about the history information. To address this challenge, we analyze a Markov persuasion process (MPP) under various information models that differ in the amount of information the receivers have about the history of the process. 
Specifically, we formulate a general partial-information model where each receiver observes the history with an $\ell$ periods lag (for $\ell \geq 0)$. Our technical contribution starts with analyzing two benchmark models, i.e., the full-history information model (i.e., $\ell = 0$) and the no-history information model (i.e., $\ell = \infty$). We establish an ordering of  the sender's payoff as a function of the informativeness of agent's information model (with no-history as the least informative),  and develop efficient algorithms to compute optimal solutions for these two benchmarks. For the information model with general $\ell$, we present the technical challenges in finding an optimal signaling mechanism, where even determining the right dependency on the history becomes difficult. Restricting the dependence on the history to a given length, we formulate the sender's problem as a bilinear optimization program. To bypass the resulting computational complexity, we use a robustness framework to design a ``simple'' \emph{history-independent} signaling mechanism that approximately achieves  optimal payoff when $\ell$ is reasonably large.
 \end{abstract}


%% file: introduction.tex
\section{Introduction}

Many platform services and markets involve {\em freelance\/} service
providers (drivers in ride-hailing markets, hosts in accommodation
services, etc.) who make voluntary decisions on when and where to
provide their services, at what quality, and at which price. Often,
the participants of these platforms lack all the necessary information
about the system (overall demand, demand imbalances, etc.) to act
optimally. Given that the platform is typically better informed, many
of them provide recommendations to the participants on their actions
in the system. For example, ride-hailing platforms such as Uber and
Lyft share real-time demand information with drivers to enable them to
make repositioning decisions. In such settings, the participants'
actions not only affect their and the platform's immediate rewards,
but also impact the evolution of the system state. Given this
dependence, a central question is to understand how the platform can make
these recommendations taking into account the participants' incentives
as well as long-term objectives like welfare and revenue.

To study such settings, we consider a model of Markovian
persuasion~\citep{wu2022sequential, gan2022bayesian, ely2017beeps,  farhadi2022dynamic, lehrer2022markovian},
where a single long-lived sender seeks to persuade a stream of short-lived
receivers by sharing information about a payoff-relevant state. The
state transitions are assumed to be Markovian, where the system's next
state is fully determined (stochastically) by the current state and
the receiver's action. The state of the system is observable to the
sender but not to the receivers. In line with the
literature~\citep{kamenica2011bayesian,bergemann2019information}, we
assume that the sender {\em commits} to a signaling mechanism, which
recommends an action based on the current state and the history of the
process. The receivers are myopic, and choose an action that maximizes
their expected payoff under their posterior beliefs given the
recommendation. The goal of the sender is to maximize the long-run
average reward.

In such settings, given the underlying Markovian dynamics, the
effectiveness of persuasion is impacted by the receivers'
knowledge of the history. Past analyses of Markovian persuasion settings either assume the receivers have exogenous beliefs~\citep{gan2022bayesian,wu2022sequential}, or assume that the receivers have no information about the history~\citep{lingenbrink2019optimal,anunrojwong2022information}.\footnote{Few works assume that the receivers   observe past signals but not past states~\citep{ely2017beeps,farhadi2022dynamic,renault2017optimal,ashkenazi2022markovian}; \citet{lehrer2022markovian} allows for stochastic revelation of past states. However, all these papers study the specialized setting where the state evolves independently of the receivers' actions.} However, from a practical perspective, both these assumptions are restrictive. In particular, the participants in a platform typically have beliefs that are influenced by their past experiences therein. Furthermore, these participants are likely to have some limited information of the history. For instance, in a ride-hailing setting, a driver, in addition to knowing the typical demand patterns at different locations, may also have some stale historical information
about demand at a particular location from having dropped off a rider there earlier. In order to ensure that the driver heeds a recommendation to move to that location, a platform must take into
account the existence of such limited historical information.


In this paper, we seek to understand the sender's persuasion problem
when receivers may have limited information about the history. To do this, we
define the notion of an {\em information model\/}, which specifies how
each receiver's belief (prior to receiving a recommendation) is
related to the history of the process. In addition to the full-history
information model $\Phi_\full$ (where the receivers observe the entire history) and
the no-history information model $\Phi_\no$ (where receivers have no historical
information), we consider a sequence $\Phi_\ell$ of partial-history information
models where each receiver observes the history of the system with an $\ell$
periods lag, for some fixed $\ell \geq 1$. These partial-history information models provide lower-bounds on the sender's payoff in more complex information models, and thus serve as a standard for comparison.

Our main contributions are as follows:
\begin{enumerate}
\item \textbf{Establishing benchmarks.} We begin with the analysis of the two benchmark information models, i.e., no-history and full-history information. We prove that, under the no-history information model, the optimal signaling mechanism is history-independent, whereas in the full-history information model, the optimal mechanism depends on the current state as well as the previous state-action pair. Consequently, these characterizations  allow us to formulate the sender's persuasion problem as a succinct linear program under both information models. 

\item \textbf{Ordering and solving partial-history information models.} We then analyze the sequence of partial-history information models $\Phi_\ell$. We show that the sender's optimal payoff under any such model is less than that under the no-history information model $\Phi_\no$, but greater than that under
full-history information model $\Phi_\full$. Moreover, we show that the sender's optimal payoff increases as the lag $\ell$ increases. We then identify sufficient conditions on the model primitives that ensure that the sender's optimal payoff in the two benchmark information models are equal and hence partial information about the history on the receivers' part has no adverse impact on the sender's payoff. Nevertheless, we show that the analysis of the persuasion problem in the information models $\Phi_\ell$ presents technical intricacies that leaves open even the question of existence of an optimal signaling mechanism. Due to this, we study the sender's problem restricting attention to signaling mechanisms that only depend on a fixed length of past history. Here, we show that the sender's problem can be written as a {\em bilinear} program, whose size grows exponentially in the lag $\ell$. This suggests that solving to optimality the sender's problem in the partial-history information models can be computationally challenging as well. 
\item \textbf{Simple and approximately optimal persuasion.} Due to the complexity of solving the persuasion problem optimally under partial-history information models, we take an
  alternative approach and ask whether simple {\em
    history-independent} mechanisms can achieve approximately optimal
  payoffs while simultaneously being persuasive under limited
  historical information. Using the underlying Markovian dynamics and
  a robust persuasion approach~\citep{zu2021learning}, we answer the
  preceding question positively. In particular, we construct a
  history-independent signaling mechanism whose payoff is close to the
  optimal payoff under the no-history information model, and which is
  simultaneously persuasive in information models $\Phi_\ell$ for
  all large enough $\ell$. To obtain this construction, we prove an
  extension of the splitting lemma~\citep{aumannM95,kamenica2011bayesian} to Markovian settings.
\end{enumerate}

Our results contribute to the literature on information design and
persuasion in dynamic settings, with endogenous beliefs of the
receivers. From a theoretical perspective, our results establish the
effectiveness of simple history-independent signaling mechanisms in
such settings. Furthermore, our results highlight the importance of robustness in designing signaling mechanisms; when participants in
a platform may have limited historical information, a simple but
robust signaling mechanism can achieve good performance while being
persuasive.

%% file: lit-survey.tex
\section{Literature Survey}
Our work contributes to the study of Bayesian
persuasion~\citep{kamenica2011bayesian,bergemann2016bayes,bergemann2019information,dughmi2017algorithmic} in dynamic
settings. Specifically, our work relates to the following streams of
literature.

\textbf{Markov persuasion.} A number of papers have looked at
persuasion problems where the state evolves according to Markov chain.
We discuss a few that are close to our setting.
\citet{gan2022bayesian} study an infinite-horizon dynamic persuasion
setting where the sender can observe the payoff-relevant parameter and
recommends actions to the uninformed receivers to maximize the
sender's cumulative rewards. They consider two types of receivers (myopic and
far-sighted) and show that when the receivers are myopic,
the optimal signaling strategy can be computed in polynomial time by
solving a linear program. But in the setting where the receivers are
far-sighted, it is NP-hard to find an approximately optimal policy. A
crucial difference between our model and theirs is that in our model,
the receivers' belief is endogenously determined by the sender's
signaling mechanism, while in their model, the receiver's belief about
the external parameter is exogenous and known to the sender.

\citet{wu2022sequential} focus on a finite-horizon Markov persuasion
process where a single long-lived sender seeks to persuade a stream of
myopic receivers to maximize the cumulative rewards. The state of the
world is seen by both the sender and receiver while the uncertain
outcome that affects the transition probability is only known to the
sender. The authors use a reinforcement learning approach to design an
online learning algorithm that achieves $O(\sqrt{T})$ regret. Similar
to the work \citep{gan2022bayesian}, the receiver's prior belief about
the state is an exogenous common prior distribution. They study a
finite-horizon Markov persuasion process while we consider the
infinite time horizon.

\citet{bernasconi2022sequential} study sequential persuasion problem where the sender seeks to persuade the far-sighted receiver by sharing the payoff-relevant state. The sender can observe the realization of the state, but neither the sender nor the receiver knows the state distribution. They show that without the knowledge of the state distribution, no algorithm can be persuasive at each round with high probability. The setting is different from ours because the transition probability is common knowledge, and the receivers are myopic and short-lived in our setting. 

Also relevant to us is the recent line of work on dynamic Bayesian persuasion, where the state evolves according to a Markov chain. \citet{ely2017beeps,renault2017optimal,farhadi2022dynamic} study the dynamic Bayesian persuasion problem between two long-lived players. To reiterate, in these works, the receiver's actions do not affect the state transitions. \citet{ely2017beeps} show that the sender's optimal strategy is myopic. Namely, the sender's optimal strategy ignores the effect of the sender's signals on the receiver's future belief. A generalization is the work by \citep{ashkenazi2022markovian}. They study the dynamic persuasion problem with binary states and any finite number of actions. The authors show that the sender's optimal strategy involves only two types of distribution of induced beliefs depending on the receiver's belief at each round. \citet{renault2017optimal}, who consider a similar setting propose a greedy disclosure policy and prove that it is optimal when the initial state is sufficiently close to the invariant distribution of the Markov chain. \citet{farhadi2022dynamic} study the setting in the finite time horizon and propose a time-varying optimal strategy for the sender. \citet{lehrer2022markovian} study the setting where the sender observes the state but the receivers randomly observe it. In contrast to these works, our model assumes that the receiver's actions affect the state evolution. Other papers also study dynamic persuasion for various application contexts, such as \citep{li2021sequential,wu2021sequential,board2018competitive,orlov2020persuading,bizzotto2021dynamic,alizamir2020warning}.

As examples of Markovian persuasion where the receivers have no information about the history, \citet{lingenbrink2019optimal} study the information-sharing problem in a single-server queue offering services at a fixed price. The service provider observes the queue and shares the information with the delay-sensitive Poisson arriving customers. The authors formulate the service providers' decision problem of maximizing the revenue as an infinite linear program. A similar approach is taken by \citep{anunrojwong2022information} to study information design to manage congestion in queues. 

\textbf{Robust persuasion.} Because our proposed signaling mechanism
in the partial-history information model relies on the robust
persuasion framework, our work also relates to robust persuasion.
\citet{zu2021learning} study a repeated Bayesian persuasion problem
where neither the sender nor the receiver knows the payoff-relevant
state distribution. They propose a robust signaling mechanism that
recommends persuasive recommendations at all rounds with high
probability and achieves $O(\sqrt{T\log T})$ regret. For our
robustness results, we extend their approach to settings where the
receivers' beliefs are endogenous. 

  \citet{kosterina2022persuasion} study a persuasion setting without the common prior assumption. In particular, the sender has a known prior, whereas only the set in which the receiver's prior lies is known to the sender. The sender evaluates the expected utility under each signaling mechanism concerning the worst-case prior of the receiver. \citet{ui2022optimal} study the optimal robust public information sharing where the sender discloses public information with receivers who also acquire costly private information. The sender is uncertain about the precision and the cost of the private information. Similarly, \citet{hu2021robust} study the problem where the receiver may have exogenous private information unknown to the sender. The sender seeks to maximize her expected payoff under the worst-case payoff  across the receiver’s possible private belief distributions  and then, among them, chooses the one that maximizes the expected utility under her conjectured prior. 
\citet{dworczak2022preparing} share a similar angle. Both works focus on static persuasion models with robustness to exogenous receiver beliefs whereas our model focuses on robust persuasion with the endogenous receivers' beliefs in sequential setups. Finally, there are also studies of robust persuasion with respect to receiver payoffs (e.g., \citep{babichenko2022regret}), though these are less relevant to the present work.


%% file: model.tex
\section{Model}\label{sec:model}

Informally, we study a dynamic persuasion setting between a long-lived sender and
a stream of short-lived receivers where the underlying payoff-relevant state
evolves as a Markov persuasion process. At each time $t$, a new
receiver arrives to whom the sender, after observing the
current state, recommends an action. The
receiver then chooses an action, possibly different from the sender's recommendation, after which the state updates according to a Markov transition kernel which is common knowledge among the sender and the receivers. Each receiver seeks to maximize the expected utility with respect to her (posterior) beliefs, given the sender's recommendation and her (partial) information about the history of the process. The sender's problem, our object of investigation, is to decide how to recommend actions that maximize her long-run average payoff. We now describe this model formally.

We consider a sequential setting where at each time $t \in \integers$,
the payoff-relevant state is given by $\bar{\omega}_t \in \Omega$.
Here $\Omega$ is a finite set of states. We denote the {\em signal\/}
shared by the sender as $\bar{s}_t \in S$, and the action chosen by
the receiver as $\bar{a}_t \in A$, where again $S$ is a finite set of signals
and $A$ is a finite set of actions. (We describe how the sender shares
the signals and how the receivers choose their actions in detail
below.) The state evolution is Markovian given the receiver's action:
$\prob( \bar{\omega}_t = \omega | \bar{h}_t) = p(\omega |
\bar{\omega}_{t-1}, \bar{a}_{t-1})$ for each $\omega \in \Omega$,
where $\bar{h}_t$ denotes the history at time $t$, i.e., the infinite sequence of state, action and signals up to (but not including) time $t$.
Here, $p \colon \Omega \times \Omega \times A \to [0,1] $ is a
stationary Markovian transition kernel, with 
$p(\omega'|\omega, a)$ denoting the probability of the state
transitioning from $\bar{\omega}_{t-1} = \omega$ to
$\bar{\omega}_t = \omega'$ after the receiver takes action
$\bar{a}_{t-1} = a$. At the end of each time $t$, the corresponding
receiver obtains a payoff given by $u(\bar{\omega}_t, \bar{a}_t) \in \mathbb{R}$, whereas the sender obtains a reward given by $v(\bar{\omega}_t, \bar{a}_t) \in [0,1]$.

\subsection{Signaling Mechanisms}

We assume that at each time $t$ the sender observes the history
$\bar{h}_t$ and the current state $\bar{\omega}_t$. On the other hand,
the receiver at time $t$ does not observe the current state, but, as
we discuss later, may have some information about the history. To convey payoff-relevant information about the state at each time $t$, the sender shares a {\em private} signal $\bar{s}_t$ to the corresponding receiver. In particular, the sender commits to sharing these signals using a {\em signaling mechanism}, which in general, maps the history $\bar{h}_t$
and the state $\bar{\omega}_t$ at any time $t$ to a signal $\bar{s}_t$. However, we circumscribe the class of signaling mechanisms in the following ways. First, we restrict our attention to signaling mechanisms that depend only on a
finite part of the history at each time. While this assumption is primarily motivated by practical concerns, it also allows us to avoid some technical issues in defining the sender's long-run average payoff if the signaling mechanism depends on the infinite history. Second, we assume that the signal at each time $t$ depends only on the historical
state-action pairs, and not on the past signals. This assumption ensures that we do not implicitly induce dependence on the infinite history via past signals. Finally, we focus on {\em direct signaling mechanisms}~\citep{bergemann2019information} where the sender shares signals that are action recommendations, i.e., $S=A$. 

Given our assumption that signals are private, it follows by the revelation principle~\citep{ely2017beeps} that considering direct signaling mechanisms is without loss of generality. Further, it is sufficient to restrict our attention to direct signaling mechanisms that are {\em persuasive}, i.e., ones where the action recommendations are optimally adopted by the receivers. In such settings, the information in past signals is already contained in the past actions, and hence the assumption that the signals only depend on past state-action pairs is not restrictive. Thus, the main restrictive assumption we make is that the signals only have finite history dependence.

Before formalizing the preceding discussion, we introduce some notation to simplify some cumbersome expressions. We let $\pair = \Omega \times A$ denote the set of state-action pairs, and we denote a generic element of $\pair$ by $x = (\omega, a)$.  Thus, $\bar{x}_t = (\bar{\omega}_t, \bar{a}_t) \in \pair$ denotes the state-action pair at time $t$, and $p(\omega'|x)$ with $x = (\omega, a)$ stands for $p(\omega'|\omega, a)$. Next, for any $k \geq 1$ and at any time $t$, a {\em slice} of history $\bar{h}_t^k$ of length $k$ describes the sequence of states-action pairs in the past $k$ time periods:  $\bar{h}_t^k = (\bar{x}_{t-k}, \dots, \bar{x}_{t-1}) \in \pair^k$. We denote a generic element of $\pair^k$ by $h^k = (x_{-k}, \dots, x_{-1})$. Finally, we let $\pair^0$ denote the singleton set consisting of the unique (empty) slice of history of length zero.

A signaling mechanism is a mapping $\sigma \colon \pair^k \times \Omega \to \Delta(A)$ (for some $k \geq 0$) that specifies for each $h^k \in \pair^k$ and $\omega \in \Omega$, the probability $\sigma(a | h^k, \omega)$ with which the sender shares the signal $\bar{s}_t = a \in A$ if the (slice of) history is $\bar{h}_t^k = h^k$ and the current state is $\bar{\omega}_t = \omega$. We let $\Sigma_k$ denote the set of all signaling mechanisms that depend only on history slices of length $k$, and let $\Sigma = \cup_{k \geq 0} \Sigma_k$. The set $\Sigma_0$ contains the signaling mechanisms that do not depend on the history.

\subsection{Beliefs and Persuasiveness}

Next, we describe the notion of persuasiveness as applied to signaling mechanisms. To do this, we need to model the receivers' beliefs about the history of the process, which in general depends endogenously on how much information they have about the past.  We capture this endogenous level of historical information through the concept of an {\em information model} (see Section~\ref{sec:endogenous}). However, to develop our concepts, we will initially consider the receiver's prior beliefs as exogenously specified.

Suppose the sender commits to a signaling mechanism $\sigma \in \Sigma_k$ for some $k \geq 0$. Fix a time $t$, and let the corresponding receiver's belief over the history $\bar{h}_t$ and the current state $\bar{\omega}_t$ (prior to receiving any signal) be denoted by $\phi_t$.  Then, upon receiving an action recommendation $\bar{s}_t = a$, the receiver's posterior belief that $\bar{\omega}_t = \omega$ can be found using Bayes' rule as 
\begin{align*}
F( \omega| a; \phi_t, \sigma) = \frac{\sum_{h^k}\phi_t(h^k, \omega) \sigma(a|h^k, \omega)}{\sum_{\omega'} \sum_{h^k}\phi_t(h^k,\omega') \sigma(a| h^k, \omega')}.
\end{align*}
Here, $\phi_t(h^k, \omega)$ denotes the receiver's marginal belief that the history slice of length $k$ is $\bar{h}_t^k = h^k \in \pair^k$ and the state is $\bar{\omega}_t = \omega$. The receiver  then chooses an action that maximizes their expected utility under their posterior belief $F(\cdot|a; \phi_t, \sigma)$. We say the signaling mechanism $\sigma$ is persuasive w.r.t.~the belief $\phi_t$, if the recommended action $\bar{s}_t = a$ is optimal for the receiver, i.e., the following inequality holds:
\begin{align*}
    \sum_{\omega}  F(\omega| a; \phi_t,  \sigma) \ {\partial u}(\omega, a, a') \geq 0, \text{ for all $a, a' \in A$},
\end{align*}
where ${\partial u}(\omega, a, a') \defeq u(\omega,a) - u(\omega, a')$
denotes the incremental payoff for the receiver for choosing action
$a \in A$ over action $a' \in A$ at state $\omega \in \Omega$. The inequality states that the receiver's expected utility with the action $a$ is higher than that with $a'$ when action $a$ is recommended.

More generally, let $\Phi = \{\phi_t : t \in \integers\}$ denote the sequence of receivers' beliefs at each time $t \in \integers$. For any such sequence $\Phi$, the set $\pers(\Phi)$ of persuasive signaling mechanisms  contains all signaling mechanisms $\sigma$ that are persuasive w.r.t.~$\phi_t$ for each $t \in \integers$. We note that the set $\pers(\Phi)$ is non-empty, since the mechanism that recommends the receivers' preferred action at each state is persuasive for sequence $\Phi$.

\subsection{Invariant Distribution}

As a step towards describing the models of endogenous historical information held by the receivers, we next analyze the induced dynamics under a signaling mechanism to characterize its invariant distribution. Suppose the sender chooses a signaling mechanism $\sigma \in \pers(\Phi) \cap \Sigma_k$ for some $k \geq 0$ and that the receivers follow the sender's recommendations. For $k \geq 1$, the induced process dynamics can be described as a Markov chain with states given by slices $\bar{h}_t^k \in \pair^k$. An invariant distribution $\pi \in \Delta(\pair^k)$ of this chain
satisfies the following balance equations:
\begin{align}\label{eq:balance}
  \sum_{x_{-k} \in \pair} \pi(x_{-k}, h^{k-1}) p(\omega| x_{-1}) \sigma(a|(x_{-k}, h^{k-1}), \omega) &= \pi(h^{k-1}, \omega, a),
\end{align}
for each $h^{k-1} = (x_{-(k-1)}, \dots, x_{-1}) \in \pair^{k-1}$ and $(\omega, a) \in \pair$. Here, the left-hand side expression gives the probability that the
slice $\bar{h}_{t+1}^k$ equals $(h^{k-1}, \omega, a)$ after a Markovian transition if the slice $\bar{h}^k_{t} =(x_{-k}, h^{k-1})$ is distributed as
$\pi$ and the receiver at time $t$ follows the sender's recommendation. The equality then just states the fact that for an invariant distribution, this distribution must be $\pi$ itself. For $k=0$, the induced process dynamics can be described as Markov chain with states $(\bar{\omega}_{t-1},\bar{a}_{t-1}) = \bar{h}_t^1 \in \pair$, and the balance equation for an invariant distribution $\pi \in \Delta(\pair)$ given by
\begin{align*}
    \sum_{x_{-1} \in \pair} \pi(x_{-1}) p(\omega | x_{-1}) \sigma(a|\omega) = \pi(\omega, a), \quad \text{for all $(\omega, a) \in \pair$.}
\end{align*}

Since the state of the induced Markov chain includes the receivers' actions, in general there might be multiple invariant distributions $\pi$ corresponding to a signaling mechanism. (As a trivial example, consider a setting with $\Omega = \{0\}$, $A = \{0,1\}$ and a receiver who is indifferent between the two actions. Let $\sigma \in \Sigma_1$ be a signaling mechanism that sends signal $\bar{s}_t = 0$ if $(\bar{\omega}_{t-1}, \bar{a}_{t-1}, \bar{\omega}_{t}) = (0,0,0)$ and sends signal $\bar{s}_t = 1$ if $(\bar{\omega}_{t-1}, \bar{a}_{t-1}, \bar{\omega}_{t}) = (0,1,0)$. Then, any distribution over $\pair$ is an invariant distribution under $\sigma$.) Hereafter, in cases where there are multiple invariant distributions, we focus on the one under which the sender's expected reward is maximized (with ties broken arbitrarily). We denote such a distribution by $\steady(\sigma)$. Note that this assumption is aligned with the notion of sender-preferred equilibrium common in the persuasion literature~\citep{kamenica2011bayesian}. 

Below, we abuse the notation slightly by letting $\pi = \steady(\sigma)$ also denote the distribution of the Markov process induced under a signaling mechanism $\sigma \in \Sigma_k$, i.e., the distribution of the entire history $\bar{h}_t$ at each time $t$. Furthermore, for any $\ell \geq 1$, we let $\pi(h^\ell)$  denote the (marginal) distribution of a slice of history $\bar{h}_t^\ell$.  

\subsection{Modeling Receivers' Endogenous Information}\label{sec:endogenous}
We now formally describe the notion of an {\em information model}, which captures the receivers'   endogenous information about the historical evolution of the process.  We consider two benchmark settings, one where each receiver fully observes the history, and the other where the receivers have no information about the history. In addition, we consider a sequence of settings where the receivers have partial information about the history. 

In general, when receivers have information about the history, the belief sequence $\Phi = \{ \phi_t : t \in \integers\}$ itself will depend on the process. The nature of this dependence is determined by the amount of information the receivers have about the past.

\begin{enumerate}
\item \textbf{Full-history information model:} To motivate the notion, fix a signaling mechanism $\sigma \in \Sigma$ and consider first the setting where at each time $t$, the corresponding receiver has complete knowledge of the history $\bar{h}_t$. Then, the receiver's belief $\phi_t$ over $(\bar{h}_t, \bar{\omega}_t)$ must put all its weight on the realized value of $\bar{h}_t$. In other words, we have for all $t \in \integers$, 
\begin{align*}
  \prob^\sigma\!\left( \phi_t = e_{h} \otimes p(\cdot | x_{-1}) \left|\ \bar{h}_t = h, \bar{\omega}_t = \omega \right.\right) =
  1, \quad \text{for all $h \in \pair^\infty$ and $\omega \in \Omega$},
\end{align*}
where $e_{h}$ is the distribution that puts all its weight on $h = (\dots, x_{-2}, x_{-1}) \in \pair^\infty$, and $e_h \otimes p(\cdot|x_{-1})$ encodes the fact that the receivers' belief about $\bar{\omega}_t$ comes from the resulting Markovian transition $p(\cdot | x_{-1})$. (Here, $\prob^\sigma$ denotes the probability measure induced by the signaling mechanism $\sigma$
together with the underlying Markovian dynamics, assuming that the receivers adopt the sender's recommendations.) When the preceding condition holds, we denote the resulting belief sequence $\{\phi_t: t  \in \integers\}$ by $\Phi_\full$ and call it the {\em full-history information model}.

\item \textbf{No-history information model:} At the other extreme, consider the case where the receivers have no information about the history of the process. Then, at any time $t$, the receiver's belief $\phi_t$ must be independent of the realized history. A natural approach, motivated by the requirement of consistency,\footnote{In certain cases, this modeling assumption can be established formally. For instance, if time periods denote the
Poisson arrival times of receivers to a stochastic system, then the receivers 
observe the system distributed as the time-average~\citep{wolff1982poisson}, which equals the
expectation w.r.t.~the invariant distribution when the latter is
unique.} is to let each belief $\phi_t$ equal the invariant distribution $\steady(\sigma)$. Specifically, we have for each $t \in \integers$, 
\begin{align*}
  \prob^\sigma\!\left( \phi_t = \steady(\sigma) \otimes P  \left|\ \bar{h}_t= h , \bar{\omega}_t = \omega\right.\right) = 1,  \text{for all $h \in \pair^\infty$ and $\omega \in \Omega$}.  
\end{align*}
Here, $\steady(\sigma) \otimes P$ encodes the distribution of $(\bar{h}_t, \bar{\omega}_{t})$ where the history $\bar{h}_t$ is distributed as $\steady(\sigma)$, and the state $\bar{\omega}_t$ is obtained from a subsequent transition from the Markov kernel $P$. For the setting where the preceding condition holds, we denote the belief sequence $\{\phi_t: t  \in \integers\}$ by $\Phi_\no$ and call it the {\em no-history information model}.

\item \textbf{Partial-history information models:} Between the
two extremes described above lie a multitude of information
models where receivers possess partial information about the process history. In such partial-history models, the belief sequence $\phi_t$ would have a complex dependence on the history $\bar{h}_t$. Although a comprehensive analysis of all such models is beyond the scope of this paper, we focus on a particular sequence of information models to capture realistic scenarios where the receivers may have some stale information about the process.\footnote{Such stale information about the process could plausibly arise from the receivers having interacted with the process in the past; however, we do not consider such repeated interactions in our model.} 

Specifically, for a fixed $\ell \geq 0$, consider the setting where the receivers observe the process with an $\ell$-period lag. In other words, at each time $t$, the receiver observes the history $\bar{h}_{t-\ell}$, i.e., all the state-action pairs before time $t-\ell$. Then, we have for each $t\in \integers$, 
\begin{align*}
    \prob^\sigma\!\left( \phi_t = e_{h_{-\ell}} \otimes P_\sigma^\ell \otimes P \left|\ \bar{h}_t = h, \bar{\omega}_t = \omega \right.\right) = 1, \text{for all $h \in \pair^\infty$ and $\omega \in \Omega$.}
\end{align*}
Here, $e_{h_{-\ell}}$ is the distribution that puts all its weight on the realization $\bar{h}_{t-\ell} = h_{-\ell}$, $P_\sigma^\ell$ encodes the subsequent $\ell$ transitions of the process, i.e., the distribution of $(\bar{x}_{t-\ell}, \dots, \bar{x}_{t-1})$ under the signaling mechanism $\sigma$, and finally, the kernel $P$ captures the subsequent distribution of the $\bar{\omega}_t$. When the preceding holds, we denote the resulting belief sequence $\{\phi_t: t \in \integers\}$ as $\Phi_\ell$, and call it the {\em partial-history information model with lag $\ell$.} We note that $\Phi_0$ is same as the full-history information model $\Phi_\full$.  

An advantage of studying the sequence $\{\Phi_\ell\}_{\ell \geq 0}$ of information models is that they serve as a standard of comparison for other more complex information models. In particular, one can show that the sender's payoff under the information model $\Phi_\ell$ acts as a lower-bound on her payoff in settings where the receivers only have limited, but arbitrary, information about states and action $\ell$ periods and further back.  Thus, while we do not capture all possible partial-history information models, our choice provides a lower bound of many other information models and gives insight into the problem's fundamental difficulty.

\end{enumerate}

\subsection{Sender's Persuasion Problem}
\label{sec:sender-problem}
Finally, we are ready to formally describe the sender's persuasion problem. We focus on settings where the sender seeks to maximize the long-run average reward over the infinite horizon. Given the Markovian state-evolution, this is equivalent to the sender choosing a signaling mechanism to
maximize the expected rewards under the resulting invariant distribution. Formally, we denote the sender's problem under the information model $\Phi$ as
\begin{align}\label{info-model}
  \mpp(\Phi) \defeq \max_{\sigma, \pi}\qquad & \expec^{\pi} [v(\omega, a)]\notag \\
  \text{ subject to, } \quad &\sigma \in \pers(\Phi)\cap \Sigma, \quad  \pi = \steady(\sigma),
\end{align}     
and let $\opt(\Phi)$ denote its optimal value. Furthermore, for $k \geq 0$, we analogously define $\mpp(\Phi, \Sigma_k)$  (and $\opt(\Phi, \Sigma_k)$) as the sender's problem (and its optimal value) when the signaling mechanism is restricted to lie in the set $\Sigma_k$. In the preceding optimization
problem, unlike a static persuasion problem, the expectation in the objective is taken with respect to the invariant distribution $\pi$ which is in turn determined by the signaling mechanism $\sigma$. 

Hereafter, we make the following standard {\em unichain} assumption~\citep{puterman2014markov,tsitsiklis}, which is common in the analysis of average-reward Markov decision processes. To state formally, a stationary Markovian policy is a decision rule that chooses a possibly randomized action based solely on the current state. Such a policy naturally induces a Markov chain over the state space. The unichain condition requires the induced Markov chain to have a single ergodic class. 
\begin{assumption}[Unichain]\label{as:unichain} Under any stationary Markovian policy,
  the resulting Markov chain has a single ergodic class, i.e., it is aperiodic and irreducible.
\end{assumption}
This assumption ensures that the invariant distribution under any signaling mechanism $\sigma \in \Sigma_0$, assuming the receivers adopt the recommendations, is unique, and thus the long-run averages are independent of the initial conditions.


%% file: benchmarks.tex
\section{Benchmarking Markovian Persuasion with Historical Information}
\label{sec:benchmark}
With the goal towards studying the sender's persuasion problem in general information models, we first analyze the sender's problem~\eqref{info-model} under the benchmark full-history and no-history information models. As we show later, the sender's optimal payoff in the two benchmark models provide bounds on the sender's optimal payoff in partial-history information models. Moreover, the results here set the stage for our subsequent analysis of the partial-history information models.

\subsection{Analysis of the Benchmark Information Models}\label{sec:analysis-benchmark}

Our analysis of the benchmark information models begins with the following lemma, which establishes that in each case, there exists an optimal signaling mechanism that is fairly simple, and does not heavily depend on the history. In particular, the optimal
mechanism under the no-history information model $\Phi_\no$ is history-independent, whereas it additionally depends on the previous state-action pair under the full-history information model $\Phi_\full$.

\begin{lemma}\label{lem:no-full-reduction}
  In the no-history information model $\Phi_\no$, there exists an
  optimal signaling mechanism $\sigma$ that is history-independent,
  i.e., $\sigma \in \Sigma_0$. Similarly, under the full-history
  information model $\Phi_\full$, there exists an optimal signaling
  mechanism $\sigma \in \Sigma_1$, which depends only on the current
  state and the previous state-action pair. 
\end{lemma}
The proof uses the underlying Markovian dynamics of the process, and is provided in Appendix~\ref{ap:endogenous} (as are the proofs are of all results in this section). Most importantly, the lemma allows us to show that the sender's problem~\eqref{info-model} in the benchmark settings can be formulated as a polynomially-sized linear program. This LP formulation plays a key role in particular in Proposition~\ref{prop:full-equals-no}, where we characterize sufficient conditions under which the sender's optimal payoffs in the two benchmark information models are equal.

To obtain the LP formulations, we begin by recalling that under the information model $\Phi_\no$ and with a signaling mechanism $\sigma \in \Sigma_0$, the states $\bar{\omega}_t \in \Omega$ form a Markov chain. On the other hand, for a signaling mechanism $\sigma \in \Sigma_1$ under the model $\Phi_\full$, the induced Markov chain can be described with states $(\bar{x}_{t-1}, \bar{\omega}_t) \in \pair \times \Omega$. We now introduce some notation to unify the presentation. First, define $\pair_\full \defeq \pair$ and $\pair_\no = \{ \star\}$. For $i \in \{\full, \no\}$, define the corresponding state space $\mathcal{W}_i \defeq \pair_i \times \Omega$. Let $v(w, a) = v(\omega, a)$ for $w = (x, \omega) \in \mathcal{W}_i$, and extend $u(w, a)$ and ${\partial u}(w, a,a')$ for $w \in \mathcal{W}_i$ similarly. Next, for $w, \hat{w} \in \mathcal{W}_i$, with $w = (x_{-1}, \omega)$ and $\hat{w} = (\hat{x}_{-1},\hat{\omega})$, we define the transition kernel $p(w | \hat{w}, \hat{a}) \defeq p(\omega| \hat{\omega}, \hat{a}) \ind\{ x_{-1} = (\hat{\omega}, \hat{a}) \text{ or } x_{-1} = \star \}$. Finally, for $x \in \pair_i$ and $w = (x_{-1}, \omega) \in \mathcal{W}_i$, define $D(x, w) \defeq \ind\{ x = x_{-1}\}$.  

With these notation in place, for $i \in \{\full, \no\}$, we consider the following linear program $\mathsf{LP}(i)$  with variables $z(w,a)$ with $w \in \mathcal{W}_i$ and $a \in A$. 
\begin{align}
  \mathsf{LP}(i) \defeq \max_{z \geq 0} \quad   \sum_{w \in \mathcal{W}_i} \sum_{a \in A } z(w, a) v(w, a) & \notag\\  
  \sum_{w \in \mathcal{W}_i} z(w , a) D(x, w) {\partial u}(w, a, a')  &\geq 0, \quad \text{for all $a, a' \in A$ and $x \in \pair_i$} \notag\\
  \sum_{\hat{w} \in \mathcal{W}_i} \sum_{\hat{a}\in A} z(\hat{w}, \hat{a}) p( w| \hat{w}, \hat{a}) &= \sum_{a} z(w,a), \quad \text{ for all $w \in \mathcal{W}_i$} \notag\\ 
  \sum_{w \in \mathcal{W}_i} \sum_{a} z(w, a)  &= 1.\label{lp:benchmark}
\end{align}
To interpret the linear program $\mathsf{LP}(i)$, we focus on the full-history information model $\Phi_\full$, and consider a persuasive signaling mechanism $\sigma \in \Sigma_1$. Then, writing the balance equations~\eqref{eq:balance} for the invariant distribution $\pi \in \steady(\sigma)$, we obtain
\begin{align*}
    \sum_{\omega_{-1}, a_{-1}} \pi(\omega_{-1}, a_{-1}) p(\omega| \omega_{-1}, a_{-1}) \sigma(a| \omega_{-1}, a_{-1}, \omega) = \pi(\omega, a).
\end{align*}
Since this equation is bilinear in $\pi$ and $\sigma$, we introduce the variables $z(w, a)$ for $w =(\omega_{-1},a_{-1}, \omega)$ to denote the summands on the left-hand side of the preceding equation, and note that $z$ constitutes the joint distribution of two consecutive state-action pairs $(\bar{x}_{t-1}, \bar{x}_{t})$ under the signaling mechanism $\sigma$. These variables are readily seen to satisfy the two equalities in $\mathsf{LP}(\full)$. Assuming $\sigma$ is persuasive under $\Phi_\full$ then yields the inequality. Thus, the variables $z(w, a)$ defined above are feasible for $\mathsf{LP}(\full)$. We obtain the following proposition upon showing the converse, i.e., for any $z$ feasible for $\mathsf{LP}(\full)$, there exists a signaling mechanism $\sigma \in \Sigma_1$ persuasive under the full-history information model $\Phi_\full$ satisfying the equation $z(w, a) = \pi(\omega_{-1}, a_{-1}) p(\omega|\omega_{-1}, a_{-1}) \sigma(a| \omega_{-1}, a_{-1}, \omega)$ (and a similar statement for the $\Phi_\no$ case). 
\begin{proposition}\label{prop:lp-formulations}
For $i \in \{\full, \no\}$, the sender's problem  $\mpp(\Phi_i)$ can be equivalently formulated as the corresponding linear program $\mathsf{LP}(i)$. In particular, for any optimal solution $z_i^*$ of $\mathsf{LP}(i)$ the signaling mechanism $\sigma_i$, defined as,  
\begin{align*}
    \sigma_i(a|w) \defeq \frac{z_i^*(w,a)}{\sum_{a'} z_i^*(w, a')}, \quad \text{for all $w \in \mathcal{W}_i$ and $a \in A$,}
\end{align*}
(if the denominator is positive; otherwise, recommending  the receivers' preferred action), is optimal for the problem $\mpp(\Phi_i)$.
\end{proposition}

Note that the linear program~$\mathsf{LP}(i)$ has $\bigO(|\mathcal{W}_i|\cdot|A|)$ variables and $\bigO(|A|^2\cdot |\pair_i| + |\mathcal{W}_i|)$ constraints. Thus, for the no-history information model $\Phi_\no$, the linear program $\mathsf{LP}(\no)$ has $\bigO(|\Omega| \cdot |A|)$ variables and $\bigO(|A|^2 +|\Omega|)$ constraints, whereas for the full-history information model $\Phi_\full$, the linear program $\mathsf{LP}(\full)$  has $O(|\Omega|^2\cdot |A|^2)$ variables and $O(|\Omega|^2 \cdot |A|^3)$ constraints. Together, the preceding result establishes that, not only the sender's problem in the benchmark models has a simple LP formulation, but also that they can be solved efficiently. As we discuss later in Section~\ref{sec:intricacies}, this is in stark contrast to the case under partial-history information models.


\subsection{Ordering and Bounding Partial-history Information Models}
\label{sec:analysis-partial-info}
Our next result justifies our choice of the two benchmarks, by showing that there is a natural nested order relating the different information models, with the two benchmark models occupying the extremes. 
\begin{lemma}\label{lem:opt}
  For $\ell \geq 0$, we have $\pers(\Phi_\full) \subseteq \pers(\Phi_\ell) \subseteq \pers(\Phi_{\ell+1}) \subseteq \pers(\Phi_\no)$ and consequently, $\opt(\Phi_\full) \leq \opt(\Phi_\ell) \leq \opt(\Phi_{\ell+1}) \leq  \opt(\Phi_\no)$. 
\end{lemma}
Intuitively, the result follows from the fact that with less information available to the receivers, the sender's ability to persuade them improves. Formally, this result is established by
showing, e.g., that any signaling mechanism that is persuasive under the model $\Phi_\full$ remains persuasive under the model $\Phi_\no$, because the sender can always share additional historical information if needed. Thus, the result implies a trade-off: by choosing the optimal signaling mechanism for the model $\Phi_\full$, the sender can simultaneously be
persuasive for all the partial-history information models $\Phi_{\ell}$, but at the
cost of lower payoffs. We illustrate the magnitude of this trade-off in the following example.
\begin{example}\label{ex:opt}
Consider a setting with $\Omega = \{0, 1\}$ and $A = \{0,1\}$. 
The receivers' utility is given by $u(\omega, a) = \ind \{\omega=a\}$, i.e., the receiver desires to match the action with the state. The sender strictly prefers the receiver choosing
action $a =1 $ over action $a=0$ in all states, i.e., $v(\omega, a)
= \ind \{a=1\}$ for all $\omega$. The transition probabilities are such that when taking action $a=0$, the state remains the same with probability $0.8$ and switches with probability $0.2$, whereas when taking action $a=1$, the state switches with probability $0.8$ and stays the same with probability $0.8$. By solving the LP formulations in the preceding section, we find that the sender's optimal payoff in the no-history information model equals $\opt(\Phi_\no) = 1$, i.e., when the receivers have no historical information, the sender can persuade the receivers to always choose her preferred action $a=1$. On the other hand, when the receivers can observe the complete history, the sender obtains a strictly lower payoff, namely $\opt(\Phi_\full) = 0.52$.
\end{example}

Thus, the example shows that, in general, the sender's optimal payoff significantly depends on the level of historical information the receivers possess. A natural question then is whether there are conditions under which historical information does not affect the sender's ability to persuade the receivers. The following proposition characterizes one such sufficient condition.

To state the result, we need some definitions. Let $\sigma \in \Sigma_0$ denote an optimal signaling mechanism in the no-history information model $\Phi_\no$. Note that under $\Phi_\no$, the receivers' prior beliefs equal the invariant distribution $\pi = \steady(\sigma)$. For any action $a \in A$ recommended by the optimal signaling mechanism, let $\mu_a$ denote the resulting posterior belief of such a receiver, and let $\mathcal{B}_\no$ denote the set of all posterior beliefs so induced. Finally, let $\conv(\mathcal{B}_\no)$ denote the convex hull of $\mathcal{B}_\no$.   

\begin{proposition}\label{prop:full-equals-no}
Suppose the set of posterior beliefs $\mathcal{B}_\no$ induced by the optimal signaling mechanism in the no-history information model $\Phi_\no$ is linearly independent. Furthermore, suppose the transition kernels lie in the convex hull of these beliefs, i.e., we have $p(\cdot | x) \in \conv(\mathcal{B}_\no)$ for all $x\in \pair$. Then, we have $\opt(\Phi_\no)=\opt(\Phi_\full)$.
\end{proposition}

The proposition, together with Lemma~\ref{lem:opt}, implies that when the conditions in the proposition statement hold, the sender can achieve the same optimal payoffs no matter the level of historical information possessed by the receivers. In particular, we have $\opt(\Phi_\ell) =\opt(\Phi_\no)$ for all $\ell \geq 0$. The proof 
of the proposition starts by writing the transition probabilities $p(\omega| x)$ as a convex combination $\sum_{a\in A} \lambda(a|x)\mu_a(\omega)$ of the beliefs in $\mathcal{B}_\no$. Then, we use these weights $\lambda(a|x)$ to explicitly construct a signaling mechanism $\widehat{\sigma} \in \Sigma_1$ which is persuasive for the model $\Phi_\full$, and induces the same set $\mathcal{B}_\no$ of posterior beliefs for the receivers with the same distribution, resulting in the same payoff for the sender. Finally, observe that the beliefs $\mu_a \in \mathcal{B}_\no$ can be easily computed from the optimal solution of the linear program~$\mathsf{LP}(\no)$; thus, the sufficient conditions in the proposition statement are straightforward to verify.

%% file: robustness.tex
\section{Optimal Persuasion in Partial-history Information Models via Robustness }\label{sec:robustness}

We now turn to the study of optimal persuasion in the general partial-history information model. We first discuss the technical intricacies in finding an optimal signaling mechanism for $\mpp(\Phi_\ell)$, and the associated computational challenges for the problem $\mpp(\Phi_\ell, \Sigma_k)$, which we formulate as a bilinear optimization program. Given these challenges, we design an approximately optimal signaling mechanism for $\mpp(\Phi_\ell)$ for large enough $\ell$, that is ``simple'' in the sense that it is history-independent and computationally efficient. Our key idea is to  leverage the fast mixing property of underlying Markov chains, whereby after sufficiently many transitions, the state distribution will be close to, though not exactly the same as, the invariant distribution. To guarantee persuasiveness for this distribution, it suffices for our design to simply guarantee robust persuasiveness for every belief that is close to the invariant distribution. We show that such robust persuasiveness can be employed to yield a simple and approximately optimal persuasion signaling mechanism for reasonably large $\ell$.

\subsection{Intricacies of Persuasion in Partial-history Information Models}\label{sec:intricacies}

Consider the sender's persuasion problem $\mpp(\Phi_\ell)$ in the partial-history information model $\Phi_\ell$ for general $\ell \geq 1$. In these models, the receivers neither have complete information about the history, nor do they completely lack history information. As we show next, this intermediate level of historical information makes the sender's persuasion problem challenging and technically intricate. In fact, even determining the degree of history dependence of the optimal signaling mechanism is difficult. This intricacy presents itself even in the simplest partial-history information model, namely $\Phi_1$, as we explain next. 

Recall that in the model $\Phi_1$, the receivers observe the history with one-period lag, i.e., at time $t$, the corresponding receiver observes the history  $\bar{h}_{t-1}$ at time $t-1$. Thus, this receiver knows the realization of $\bar{x}_{t-2} = (\bar{\omega}_{t-2}, \bar{a}_{t-2})$, but does not know $\bar{x}_{t-1} = (\bar{\omega}_{t-1}, \bar{a}_{t-1})$ and $\bar{\omega}_t$. An initial guess then is to consider signaling mechanisms in the set $\Sigma_2$, i.e., ones that make recommendations based on $(\bar{x}_{t-2}, \bar{x}_{t-1}, \bar{\omega}_t)$.  (This comports well with the full-history information model $\Phi_\full = \Phi_0$, where the optimal signaling mechanism lies in the set $\Sigma_1$.) With such a choice of $\sigma$, after seeing $\bar{h}_{t-2} = h_{-2} = (\dots, x_{-3}, x_{-2})$,  the receiver's belief about $(\bar{x}_{t-1}, \bar{\omega}_t)$ is given by 
\begin{align*}
    \prob^\sigma\!\left( \bar{x}_{t-1} = x_{-1}, \bar{\omega}_{t} = \omega \left|\ \bar{h}_{t-2} = h_{-2}\right.\right) = p(\omega_{-1} | x_{-2}) \sigma(a_{-1} | x_{-3}, x_{-2}, \omega_{-1}) p(\omega | x_{-1}).
\end{align*}
Thus, the receiver's belief depends not just on $x_{-2}$, but also on the realization $x_{-3}$ of  $\bar{x}_{t-3}$. Requiring the signaling mechanism $\sigma$ to be persuasive for different beliefs of the receiver corresponding to different realization of $x_{-3}$, without depending on $x_{-3}$ explicitly, is unlikely to yield optimality. Consequently, one is tempted to consider signaling mechanisms $\sigma \in \Sigma_3$, i.e., ones that make recommendation based on $(\bar{x}_{t-3}, \bar{x}_{t-2}, \bar{x}_{t-1}, \bar{\omega}_t)$. However, a similar argument as above would imply that for such signaling mechanisms, the receiver's belief would depend on the realization of $(\bar{x}_{t-4}, \dots \bar{x}_{t-2})$. In general, for any signaling mechanism $\sigma \in \Sigma_k$ with $k \geq 1$, the receiver at time $t$ has a different belief for different realizations of $(\bar{x}_{t-(k+1)}, \dots, \bar{x}_{t-2})$, but the signaling mechanism $\sigma$ does not base its recommendation on the realization of $\bar{x}_{t-(k+1)}$. Due to this mismatch of dependencies, it is unclear what the right dependence of the optimal signaling mechanism is on the history, or for that matter, even whether there exists an optimal signaling mechanism within the class $\Sigma$ of signaling mechanisms.

Given the ambiguity regarding the degree of history dependence, one may instead consider optimizing the sender's payoff within a restricted subset $\Sigma_k$ of signaling mechanisms, for some fixed $k$. However, even this restricted problem turns out to be computationally challenging since, unlike the case for $\Phi_\full$ and $\Phi_\no$, it does not reduce to a linear program. In particular, the following result formulates the sender's problem $\mpp(\Phi_1, \Sigma_1)$ as a bilinear program.
\begin{proposition}\label{prop:partial-bilinear} The sender's problem $\mpp(\Phi_1, \Sigma_1)$ can be formulated as the following bilinear program:
\begin{align}\label{opt:non-linear}
\max_{z \geq 0} \sum_{x_{-2}, x_{-1} \in \pair} \sum_{\omega, a} & z(x_{-2}, x_{-1}, \omega, a) v(\omega, a) \notag\\    
\sum_{x_{-1} \in \pair} \sum_{\omega} z(x_{-2}, x_{-1}, \omega, a) {\partial u}(\omega, a, a') &\geq 0, \quad \text{for all $x_{-2} \in \pair$ and $a, a' \in A$.}\notag\\
\sum_{x_{-3} \in \pair} z(x_{-3}, x_{-2}, x_{-1})p(\omega|x_{-1}) &= \sum_{a \in A} z(x_{-2}, x_{-1}, \omega, a), \quad \text{for all $x_{-2}, x_{-1}\in \pair$ and $\omega \in \Omega$}\notag\\
\sum_{x_{-2}, x_{-1}, \omega, a} z(x_{-2}, x_{-1}, \omega, a) &= 1\notag\\
z(x_{-2}, x_{-1}, \omega, a)\cdot \sum_{a' \in A} z(x'_{-2}, x_{-1}, \omega, a') &= z(x'_{-2}, x_{-1}, \omega, a) \cdot \sum_{a' \in A} z(x_{-2}, x_{-1}, \omega, a'),\notag\\
&\qquad \qquad \qquad \text{for all $x'_{-2}, x_{-2}, x_{-1}, (\omega , a) \in \pair$.} 
\end{align}
\end{proposition}
To elaborate, for any $\sigma \in \Sigma_1$, assuming the receivers follow the recommendation and given our preceding discussion, the underlying process dynamics can be described as a Markov chain with states given by slices $\bar{h}_t^2 = (\bar{x}_{t-2},\bar{x}_{t-1})$. The balance equation for this chain's invariant distribution $\pi \in \steady(\sigma)$ is given by 
\begin{align*}
    \sum_{x_{-2} \in \pair} \pi(x_{-2}, x_{-1}) p(\omega|x_{-1}) \sigma( a | x_{-1}, \omega) = \pi(x_{-1}, \omega, a),
\end{align*}
for all $x_{-1} \in \pair$ and $(\omega, a) \in \pair$. As in the case for $\Phi_\full$ and $\Phi_\no$, this equation is non-linear in $\pi$ and $\sigma$. However, introducing the variables $z(x_{-2}, x_{-1}, \omega, a) \defeq \pi(x_{-2}, x_{-1}) p(\omega|x_{-1}) \sigma( a | x_{-1}, \omega)$ no longer yields a linear program, because of the restriction that $\sigma$ cannot depend on $x_{-2}$. In particular, the non-linear equality constraint in \eqref{opt:non-linear} explicitly encodes the requirement that $z(x_{-2}, x_{-1}, \omega, a)/ \sum_{a' \in A} z(x_{-2}, x_{-1}, \omega, a)$ is independent of $x_{-2}$. 

A similar argument holds for any partial-history information model $\Phi_\ell$, where for any signaling mechanism $\sigma \in \Sigma_k$ for $k \geq 0$, the belief of a receiver at time $t$ depends on the realization of $(\bar{x}_{t-\ell-k}, \dots, \bar{x}_{t-\ell-1})$, but the signaling mechanism $\sigma$ does not base its recommendation on $(\bar{x}_{t-\ell-k}, \dots, \bar{x}_{t-k-1})$. Again, the sender's problem $\mpp(\Phi_\ell, \Sigma_k)$ can be shown to be a bilinear optimization problem, whose size is exponential in $\ell+k$. We skip the details for the sake of brevity.

The preceding discussion hints at a trade-off faced by the sender in the model $\Phi_\ell$ for some $\ell \geq 1$. On one hand, the sender can adopt the optimal signaling mechanism for the full-history information model $\Phi_\full$, which is simple in that it only uses the previous state-action pair (and the current state) to recommend an action, and is persuasive for the model $\Phi_\ell$, as shown in Proposition~\ref{lem:opt}.  However, this simplicity may come at the cost of substantially lower payoffs, especially if $\ell$ is large. On the other hand, the sender may choose a large $k$ and solve a non-linear program akin to \eqref{opt:non-linear} to find the best signaling mechanism within the class $k$, which likely will yield higher payoffs, at the cost of substantial computational complexity. (See e.g., Fig~\ref{fig:exopt}.) In the following section, we provide an approach to overcome this trade-off, as long as one is satisfied with approximate optimality.

\begin{figure}[t]
\centering
\begin{tabular}{cccccccccccc}
\toprule
 $k$& & & 0 & & 1 & &  2 & & 3 & & 4 \\
 \midrule
 $\opt(\Phi_1, \Sigma_k)$ & & & $0.576$ &  & $0.772$ & &  $0.799$ & &  $0.808$ & & $0.811$\\
 \bottomrule
\end{tabular}
\caption{(Example~\ref{ex:opt} contd.) Sender's optimal payoff in $\mpp(\Phi_1, \Sigma_k)$ for different values of $k$. The optimal values are obtained by numerically solving bilinear optimization programs analogous to \eqref{opt:non-linear} for different values of $k$. Here, $\opt(\Phi_\full) = 0.52$ and $\opt(\Phi_\no) = 1$.}  
\label{fig:exopt}
\end{figure}

\subsection{Approximately Optimal Persuasion via Robustness}

In this section, we ask and answer the following questions: in partial-history information models, can ``simple'' signaling mechanisms guarantee persuasiveness without sacrificing the sender's payoff too much? And if so, can we find such a mechanism in a computationally efficient manner? To answer these questions positively, we take an approach inspired from robust persuasion~\citep{zu2021learning}. Our starting point is the observation that, for a signaling mechanism $\sigma$, if the underlying Markov chain mixes rapidly, the belief of the receiver who has stale historical information must be close to the invariant distribution $\pi =\steady(\sigma)$. Thus, if $\sigma$ is simultaneously persuasive for all distributions close to $\pi$, it must be persuasive under the information models $\Phi_\ell$ for all large enough $\ell$. Using this insight, we explicitly construct a {\em robustly persuasive} history-independent signaling mechanism with good payoff guarantees.

To begin, recall that for any history-independent signaling mechanism $\sigma \in \Sigma_0$, assuming the receivers follow the recommendation, $\bar{x}_t = (\bar{\omega}_t, \bar{a}_t) \in \pair$ forms a Markov chain. Let $\pi = \steady(\sigma)$; we abuse the notation slightly by letting $\pi$ also denote the marginal over $\bar{\omega}_t$, i.e., $\pi(\omega) = \sum_{a \in A} \pi(\omega, a)$ for $\omega \in \Omega$. For $\epsilon \geq 0$, let $\mathsf{B}_1(\pi,\epsilon)$ denote the set of all distributions $\mu \in \Delta(\Omega)$ that are $\epsilon$-close to $\pi$ in $\ell_1$-norm:  $\mathsf{B}_1(\pi,\epsilon) \defeq \{ \mu \in \Delta(\Omega) : \| \mu - \pi\|_1 \leq \epsilon\}$. 

An $\epsilon$-{\em robustly persuasive} signaling mechanism $\sigma \in \Sigma_0$ is one whose recommendations would be optimally adopted by any receiver whose prior belief about $\bar{\omega}_{t}$ lies in the set $\mathsf{B}_1(\pi,\epsilon)$:
\begin{align*}
  \sum_{\omega} \mu(\omega) \sigma(a |  \omega) {\partial u}(\omega, a, a') \geq 0, \text{ for all $a, a' \in A$ and all $\mu \in \mathsf{B}_1(\pi,\epsilon)$,}
\end{align*}
where $\pi =\steady(\sigma)$. We denote the set of $\epsilon$-robustly persuasive signaling mechanisms by $\rp(\epsilon)$. The value of $\epsilon$ captures the degree of robustness of a mechanism $\sigma \in \rp(\epsilon)$, with smaller values corresponding to lower robustness. Observe that for all $\epsilon \geq 0$, we have $\rp(\epsilon) \subseteq \pers(\Phi_\no) \cap \Sigma_0$, with equality for $\epsilon=0$. Furthermore, the set $\rp(\epsilon)$ is non-empty for all $\epsilon \geq 0$, as it contains the signaling mechanism that recommends an receiver-optimal action at each state.

Our next result describes the relation between $\rp(\epsilon)$ and the set $\pers(\Phi_\ell)$ for large $\ell$. For $\sigma \in \Sigma_0$ and $\ell \geq 1$, let $Q_\sigma^\ell(x, \omega) \defeq \prob^\sigma\!\left( \bar{\omega}_{\ell} = \omega | \bar{x}_{-1} = x\right)$ denote the distribution of $\bar{\omega}_\ell$ under $\sigma$, given $\bar{x}_{-1} = x \in \pair$. Define $d_\ell(\sigma)$ as the maximum $\ell_1$-distance between $Q_\sigma^\ell(x)$ and $\pi = \steady(\sigma)$ over $x \in\pair$:
\begin{align*}
    d_\ell(\sigma) \defeq \sup_{x \in \pair} \left\| Q^\ell_\sigma(x) - \pi\right\|_1 = \sup_{x \in \pair} \sum_{\omega} \left| Q^\ell_\sigma(x, \omega) - \pi(\omega)\right|.
\end{align*}
Finally, let $\gamma_\star(\sigma)$ denote the absolute spectral gap~\citep{levin2017markov} of the Markov chain $\{\omega_t\}$ under $\sigma$ and $\pi_{\min}(\sigma) = \min_{\omega} \pi(\omega) > 0$. We have the following result.
\begin{lemma}\label{lem:robust-partial}
    Suppose the signaling mechanism $\sigma \in \Sigma_0$ is $\epsilon$-robustly persuasive for $\epsilon > 0$. If $\ell \geq 0$ satisfies $d_\ell(\sigma) \leq \epsilon$, then $\sigma \in \pers(\Phi_\ell)$. In particular, $\sigma \in \pers(\Phi_\ell)$ for all $\ell \geq \frac{1}{\gamma_\star(\sigma)} \log\left(\frac{2}{\epsilon \pi_{\min}(\sigma)}\right)$.
\end{lemma} 
The proof of the bound in the lemma statement uses the unichain assumption (Assumption~\ref{as:unichain}) to bound the mixing time of the underlying Markov chain. The result implies that in order to find a signaling mechanism in $\pers(\Phi_\ell)$, it suffices to find a history-independent signaling mechanism in the set $\rp(\epsilon)$ for small enough $\epsilon$. We highlight that the required value of $\epsilon$ decays exponentially in $\ell$, and hence the robustness requirements are not too stringent.


Given this preceding result, we seek to identify a robustly persuasive
mechanism with good guarantees on the sender's payoff. We prove such a
result next. To state the result, we need a definition. Define the
sets $\mathcal{P}_a \subseteq \Delta(\Omega)$ as follows:
\begin{align*}
  \mathcal{P}_a \defeq \left\{ \mu \in  \Delta(\Omega)  : a \in \argmax_{a'} \expec_\mu[ u(\omega, a')] \right\}.
\end{align*}
In other words, $\mathcal{P}_a$ is the set of beliefs for which the
receiver finds it optimal to choose action $a$. Similar to
\citet{zu2021learning}, we make the following regularity assumption on the
receivers' utility function.
\begin{assumption}[Regularity]\label{as:regularity}
  There exists a positive constant $D >0$ and beliefs
  $\eta_a \in \mathcal{P}_a$ for $a \in A$ such that
  $\mathsf{B}_1(\eta_a, D) \subseteq \mathcal{P}_a$ for each
  $a \in A$, where $\mathsf{B}_1(\eta, D)$ is an $\ell_1$-ball of size
  $D$ centered at $\eta$.
\end{assumption}
The regularity assumption ensures that each action for the receiver is
optimal for a set of beliefs with non-zero (Lebesgue) measure. This
ensures the exclusion of pathological instances, where there is an
action that is optimal for the receiver under a unique belief.
Furthermore, \citet{zu2021learning} establish that the regularity assumption
ensures that, in static problems, the cost of requiring robustness
scales linearly in the degree of robustness.

Next, let $a_{\omega} \in A$ be a best response for a receiver at
state $\omega \in \Omega$, i.e.,
$a_\omega \in \arg\max_{a \in A} u(\omega, a)$ for each
$\omega \in A$. Let
$P_f(\omega, \omega') \defeq p(\omega'|\omega, a_\omega)$ denote the
transition probability from state $\omega$ to state $\omega'$ on
choosing the action $a_\omega$, and let $P_f$ denote the transition
matrix of the underlying process. Note that the unichain assumption
implies that $P_f$ is ergodic. 
Let $\nu_f \in \Delta(\Omega)$ denote the steady state distribution
under the transition kernel $P_f$. Furthermore, let
$\tau \defeq \max_{\omega} 1/ \nu_f(\omega)$ denote the maximum
expected {\em first return time} across all states. Finally, let $s_f$ be the smallest positive singular value of the matrix $I -P_f$.

With these definitions in place, we are now ready to present the main
result of this section.
\begin{theorem}\label{thm:robustness-history-independent}
  For
  $\epsilon < \frac{s_f w_{\min}D}{2(s_f + 2(1+
    \tau)\sqrt{|\Omega|})}$, there exists a signaling mechanism $\widehat{\sigma} \in \rp(\epsilon)$ with
  the sender's payoff bounded below by
  \begin{align*}
    \left(1 - \frac{2 \epsilon}{w_{\min}D} \left(1 + \frac{2(1+ \tau)\sqrt{|\Omega|}}{s_f}    \right) \right) \cdot \opt(\Phi_\no),
  \end{align*}
  where $w_{\min}$ is the smallest positive probability of
  recommending an action under the optimal mechanism under $\Phi_\no$.
  
\end{theorem}

The preceding result, together with Lemma~\ref{lem:robust-partial},  implies that for the partial-history model $\Phi_\ell$ with large enough $\ell$, the sender need not solve a non-linear program. Instead, the sender can use a simple history-independent signaling mechanism to obtain approximately
optimal payoffs. The proof involves an explicit
construction of such a signaling mechanism
$\widehat{\sigma} \in \rp(\epsilon)$. From a computational
perspective, constructing such a mechanism requires solving $\mpp(\Phi_\no)$ (equivalently the linear program $\mathsf{LP}(\no)$), 
and solving a separate linear program~\eqref{small-perturbation-LP} with $\bigO(|\Omega|)$ variables and
constraints (see Lemma~\ref{lem:bound-LP} in Appendix~\ref{ap:robustness} for details). Thus, not only
the proposed mechanism obtains approximately optimal payoffs, but it
also can be computed efficiently.

To construct the mechanism $\widehat{\sigma} \in \rp(\epsilon)$ with good guarantees on the sender's payoff, we use a similar
approach as in \citep{zu2021learning}, where we first identify a set
of beliefs that we seek to induce as the receivers' posterior beliefs
under the constructed mechanism. These beliefs are chosen to lie
strictly in the interior of the sets $\mathcal{P}_a$, to ensure that
the actions remain optimal for all close-by beliefs. However, unlike
the static setting of \citep{zu2021learning}, the endogeneity of the receivers' prior
belief in our setting raises the question of whether there exists a mechanism that
induces these beliefs as posteriors. To exhibit such a mechanism, we
prove the following analog of the splitting lemma~\citep{aumannM95, kamenica2011bayesian} for the Markovian
persuasion setting, providing conditions on a set of beliefs under which a signaling mechanism exists that induces those
beliefs as posteriors.
\begin{lemma}\label{lem:splitting-lemma-markov}
  For a finite set $S$, let $\{ \mu_s : s \in S\}$ be a set of
  beliefs, and for each $s \in S$, let $a_s \in A$ be such that
  $\mu_s \in \mathcal{P}_{a_s}$. Suppose there exists a set of weights
  $\{w_s \geq 0 : s \in S\}$ such that $\sum_{s \in S} w_s = 1$
and
\begin{align*}
  \sum_{s \in S}  \sum_{\omega} w_s \mu_s(\omega) p(\cdot | \omega, a_s) = \sum_{s \in S} w_s \mu_s.
\end{align*}
Then, there exists a signaling mechanism $\sigma \in \Sigma_0$, which sends signals
$s \in S$ with probability
$\sigma(s|\omega) = \frac{w_s \mu_s(\omega)}{\sum_{s' \in S} w_{s'}
  \mu_{s'}(\omega)}$, such that under the no-history information model
$\Phi_\no$, the posterior belief of a receiver on receiving signal $s$
equals $\mu_s$.

Conversely, for any signaling mechanism $\sigma \in \pers(\Phi_\no) \cap \Sigma_0$, 
there exists weights $w_a \geq 0$ and beliefs
$\mu_a \in \mathcal{P}_a$ with $\sum_{a \in A} w_a = 1$ and
$\sum_{\omega, a} w_a \mu_a(\omega) p(\cdot | \omega, a) = \sum_a w_a
\mu_a$, such that $\sigma(a|\omega) = \frac{w_a \mu_a(\omega) }{\sum_{a'} w_{a'} \mu_{a'}(\omega)}$.
\end{lemma}
With this splitting lemma in hand, we construct our robustly
persuasive mechanism by proving the existence of weights satisfying
the preceding condition. We provide the
complete proof in Appendix~\ref{ap:robustness}.

%% file: conclusion.tex
\section{Conclusion}
We consider a Markovian persuasion setting between a single long-lived
sender and a stream of receivers, where the sender commits to a
signaling mechanism to maximize the long-run average reward. To
capture settings where the receiver may have limited historical
information, we analyze a set of endogenous information models. We
observe that the sender's persuasion problem can be posed as simple
linear programs under the full-history and the no-history information
models. However, when the receiver has partial information about the history, the sender's problem presents technical intricacies, and is computationally challenging due to its non-linear nature. To
overcome this difficulty, we adopt a robust persuasion approach to
construct a simple history-independent signaling mechanism with strong
guarantees on the payoff, that nevertheless is persuasive for all
models with sufficiently limited historical information. Furthermore, the
robust mechanism can be computed efficiently by solving simple linear
programs. From a theoretical perspective, our work highlights the
trade-off between higher sender's payoffs and being persuasive under a
larger class of information models.


We have focused on the setting where the sender seeks to maximize the
long-run average payoff. An alternative objective is to maximize the
cumulative discounted reward. However, note that in endogenous
information models, the receivers' belief is related to the invariant
distribution of the process, which equals long-run averages in
stationary models. Thus, the persuasion problem with discounted
rewards is similar to a constrained Markov decision process where the
objective involves discounting and the constraint requires averaging. Even in the classical context of constrained MDPs, problems with distinct discount factors in the objective and the
constraints are challenging (note that averaging can be interpreted as the limit where the discount factor converges to one). For instance, \citet{feinberg1994markov,
  feinberg1995constrained} show that in such settings the optimal
policy need not be stationary. An additional complexity that arises
with discounting rewards is the dependence on the initial conditions.
Given these challenges, a systematic analysis of Markov persuasion process with endogenous beliefs
and discounted rewards is an interesting direction for further
theoretical research.




%% file: appendix.tex
\section{Proofs from Section~\ref{sec:benchmark}}
\label{ap:endogenous}
\subsection{Proofs from Section~\ref{sec:analysis-benchmark}}
\label{ap:sec-benchmark}
\begin{proof}[Proof of Lemma~\ref{lem:no-full-reduction}] We prove the two statements corresponding to the no-history information model $\Phi_\no$ and the full-history information model $\Phi_\full$ separately. 

\textbf{1. No-history information model $\Phi_\no$:} We prove the statement by showing that for any $\sigma \in \Sigma_k \cap \pers(\Phi_\no)$ for some $k$, there exists a $\widehat{\sigma} \in \Sigma_0 \cap \pers(\Phi_\no)$ with the same payoff for the sender. 

First, recall that in the no-history information model $\Phi_\no$ and under the signaling mechanism $\sigma \in \Sigma_k \in \pers(\Phi_\no)$, the receiver's prior belief is given by the invariant distribution $\pi = \steady(\sigma)$. Note that $\pi$ describes the invariant distribution of the slice $\bar{h}_t^k \in \pair^k$ under $\sigma$. By abusing the notation, we let $\pi(\omega, a) \defeq  \sum_{h^{k-1}} \pi(h^{k-1}, \omega, a)$ also denote the marginal distribution of $(\bar{\omega}_t, \bar{a}_t)$ under $\pi$.


Now, define the signaling mechanism $\widehat{\sigma} \in \Sigma_0$ as follows: for $\omega \in \Omega$ with $\sum_a \pi(\omega, a) > 0$, let  
\begin{align*}
        \widehat{\sigma}(a|\omega) \defeq \frac{\pi(\omega,a)}{\sum_{a'}\pi(\omega, a')}\qquad \text{for $a \in A$}
\end{align*}
and for any $\omega \in \Omega$ with $\sum_{a \in A} \pi(\omega, a) = 0$, we let $\widehat{\sigma}$ recommend the receiver-optimal action at $\omega$. 

We first show that $\widehat{\pi} \in \Delta(\pair)$ with $\widehat{\pi}(\omega,a) \defeq \pi(\omega, a)$ is an invariant distribution under $\widehat{\sigma}$. To see this, observe that if $\sum_{a \in A} \pi(\omega, a) > 0$, we have 
\begin{align*}
    \sum_{x_{-1} \in \pair} \widehat{\pi}(x_{-1}) p(\omega | x_{-1}) \widehat{\sigma}(a|\omega) &= \left(\sum_{x_{-1} \in \pair} \pi(x_{-1}) p(\omega | x_{-1}) \right) \widehat{\sigma}(a|\omega)\\
    &= \left(\sum_{a'} \pi(\omega, a') \right) \widehat{\sigma}(a|\omega)\\
    &= \pi(\omega, a)\\
    &= \widehat{\pi}(\omega, a).
\end{align*}
Here, the first and fourth equality follows from the definition of $\widehat{\pi}$, the second follows from the fact 
that $\pi = \steady(\sigma)$, and the third equality follows from the definition of $\widehat{\sigma}$. Moreover, if 
$\sum_{a \in A} \pi(\omega, a) = 0$, then $\sum_{x_{-1} \in \pair} \pi(x_{-1}) p(\omega | x_{-1}) = 0$, and 
$\widehat{\pi}(\omega, a) = \pi(\omega, a) = 0$, and hence the equality continues to hold. Thus, $\widehat{\pi}$ 
satisfies the balance equations~\eqref{eq:balance} under $\widehat{\sigma}$, and thus, $\widehat{\pi} = \steady(\widehat{\sigma})$. 

Finally, in the information model $\Phi_\no$ and under the mechanism $\sigma_k$, a receiver's posterior belief that $\bar{\omega}_t= \omega$ after being recommended $\bar{s}_t = a$ is given by
\begin{align}
  F(\omega | a; \phi_t, \sigma) &= \frac{ \sum_{ h^k } \pi( h^k ) p(\omega | x_{-1}) \sigma(a | h^k , \omega) }{ \sum_{ h^k}\sum_{\omega'} \pi( h^k) p(\omega' | x_{-1}) \sigma(a |  h^k, \omega')} \label{eq:sender-opt-no-history-posterior}
\end{align}
where $\pi(h^k)$ denotes the probability that the slice $\bar{h}_t^k = h^k \in \pair^k$ under the invariant distribution $\pi = \steady(\sigma)$. Summing the  balance equation~\eqref{eq:balance} for $\pi$ over all values of $h^{k-1} \in \pair^{k-1}$, we obtain
\begin{align*}
    \sum_{h^k} \pi(h^k) p(\omega|x_{-1}) \sigma(a| h^k, \omega) = \sum_{h^{k-1}} \pi(h^{k-1}, \omega, a) = \pi(\omega, a).
\end{align*}
Substituting in \eqref{eq:sender-opt-no-history-posterior}, we obtain $F(\omega | a; \phi_t, \sigma) = \frac{\pi(\omega, a)}{\sum_{\omega'} \pi(\omega', a)} = \frac{\widehat{\pi}(\omega, a)}{\sum_{\omega'} \widehat{\pi}(\omega', a)} = F(\omega|a; \phi_t, \widehat{\sigma})$, where the last equality follows by a similar argument for $\widehat{\sigma}$ in the model $\Phi_\no$. Since
the receivers have the same belief under $\sigma$ and $\widehat{\sigma}$, and further $\sigma \in \pers(\Phi_\no)$, we conclude that $\widehat{\sigma} \in \pers(\Phi_\no)$. The result then follows from the fact that the sender's payoffs under the two mechanisms are equal.

\textbf{2. Full-history information model $\Phi_\full$:} The proof for the full-history information model follows along similar lines. Fix $\sigma \in \pers(\Phi_\full) \cap \Sigma_k$, and let $\pi = \steady(\sigma)$ denote its invariant distribution. As before, we define a mechanism $\widehat{\sigma} \in \Sigma_1$ and show that it is persuasive under $\Phi_\full$ and achieves the same payoff for the sender. Towards that end, let $\pi(x_{-1}, \omega, a)$ denote the marginal distribution under $\pi$ that $(\bar{x}_{t-1}, \bar{\omega}_t, \bar{a}_t) = (x_{-1}, \omega, a)$ and define for $x_{-1} \in \pair$ and $\omega \in \Omega$,
\begin{align*}
    \widehat{\sigma}(a|x_{-1}, \omega) \defeq \frac{\pi(x_{-1}, \omega, a)}{\sum_{a'} \pi(x_{-1}, \omega, a')},
\end{align*}
if the denominator is positive, and otherwise let $\widehat{\sigma}$ recommend the receiver-optimal action at $\omega$. Similarly, define $\widehat{\pi} \in \Delta(\pair)$ to be $\widehat{\pi}(\omega, a) \defeq \sum_{x_{-1}} \pi(x_{-1}, \omega, a)$. Whenever $\sum_{a'} \pi(x_{-1}, \omega,a) > 0$, we have
\begin{align*}
      \sum_{x_{-1} \in \pair} \widehat{\pi}(x_{-1}) p(\omega | x_{-1}) \widehat{\sigma}(a|x_{-1},\omega) &= \sum_{x_{-1} \in \pair} \left(\sum_{x_{-2} \in \pair} \pi(x_{-2}, x_{-1}) p(\omega|x_{-1}) \right) \widehat{\sigma}(a|x_{-1},\omega)\\
      &= \sum_{x_{-1} \in \pair} \left(\sum_{a'} \pi(x_{-1}, \omega, a')\right) \widehat{\sigma}(a|x_{-1},\omega)\\
      &= \sum_{x_{-1} \in \pair} \pi(x_{-1}, \omega, a)\\
      &= \widehat{\pi}(\omega, a).
\end{align*}
On the other hand, if $\sum_{a'} \pi(x_{-1}, \omega,a) = 0$, we obtain both sides of the equations are zero. Thus, we conclude that $\widehat{\pi} \in \steady(\widehat{\sigma})$.

Finally, in the information model $\Phi_\full$ and under the signaling mechanism $\sigma$, we have
\begin{align*}
    F(\omega| a; \phi_t, \widehat{\sigma}) = \frac{ p(\omega|x_{-1}) \sigma(a| h^k, \omega)}{\sum_{\omega'} p(\omega'|x_{-1}) \sigma(a| h^k, \omega')}.
\end{align*}
As $\sigma \in \pers(\Phi_\full)$, we obtain for any $a , a' \in A$ and all $h^k \in \pair^k$,
\begin{align*}
    \sum_{\omega} p(\omega|x_{-1}) \sigma(a| h^k, \omega) {\partial u}(\omega, a,a') &\geq 0.
\end{align*}
After multiplying by $\pi(h^k)$, summing up over $ (x_{-k}, \dots, x_{-2})$, and using the fact 
that $ \pi(x_{-1}, \omega, a) = \sum_{(x_{-k}, \dots, x_{-2})} \pi(h^k) p(\omega| x_{-1}) \sigma(a|h^k, \omega)$ from
the balance equations for $\pi = \steady(\sigma)$, we obtain for all  $x_{-1} \in \pair$ and $a, a' \in A$, 
\begin{align*}
    \sum_{\omega} \pi(x_{-1}, \omega, a) {\partial u}(\omega, a,a') &\geq 0.
\end{align*}
Now, from the definition of $\widehat{\sigma}$, we have $\pi(x_{-1}, \omega, a) = \widehat{\pi}(x_{-1}) p(\omega|x_{-1})\widehat{\sigma}(a|x_{-1}, a)$ if $\sum_{a'} \pi(x_{-1}, \omega, a') > 0$. Furthermore, under this condition and using the fact that $\pi \in \steady(\sigma)$, we have $\widehat{\pi}(x_{-1}) = \sum_{x_{-2}} \pi(x_{-2}, x_{-1}) = \sum_{x_{-2}} \pi(x_{-1}, x_{-2}) > 0$. Thus, we conclude that 
\begin{align*}
    p(\omega|x_{-1})\widehat{\sigma}(a| x_{-1}, \omega) {\partial u}(\omega, a,a') &\geq 0.
\end{align*}
On the other hand, if $\sum_{a'} \pi(x_{-1}, \omega, a') > 0$, then $\widehat{\sigma}$ recommends the receiver-optimal action. Thus, we conclude that $\widehat{\sigma} \in \pers(\Phi_\full)$. Once again, the result then follows as the sender's payoffs under the two mechanisms are equal.\qed

\end{proof}

\begin{proof}[Proof of Proposition~\ref{prop:lp-formulations}] We prove the statement for $\mpp(\Phi_\full)$.  A similar argument, with minor modifications, obtains the equivalence of $\mpp(\Phi_\no)$ and $\lp(\no)$; we omit it for brevity. 

From Lemma~\ref{lem:opt}, we know that there exists an optimal signaling mechanism for $\mpp(\Phi_\full)$ within the set $\Sigma_1$. The proof shows that for any $\sigma \in \pers(\Phi_\full) \cap \Sigma_1$, there exists a corresponding feasible solution $z$ to $\lp(\full)$ whose objective equals the sender's payoff, and conversely, for any feasible solution $z$ to $\lp(\full)$, there exists a signaling mechanism $\sigma \in \pers(\Phi_\full)$ with sender's payoff equaling the sender's payoff at $z$. 

To begin, fix $\sigma \in \pers(\Phi_\full) \cap \Sigma_1$, and let $\pi = \steady(\sigma) \in \Delta(\pair)$. The balance equations~\eqref{eq:balance} are given by
\begin{align*}
    \sum_{x_{-1}} \pi(x_{-1})p(\omega|x_{-1})\sigma(a| x_{-1}, \omega) = \pi(\omega,a), \quad \text{for all $(\omega, a) \in \pair$.}
\end{align*}
Define $z(w, a) \defeq \pi(x_{-1}) p(\omega|x_{-1})\sigma(a|x_{-1}, \omega)$ for $w=(x_{-1}, \omega) \in \mathcal{W}_\full$ and $a \in A$. It is straightforward to check that $z$ satisfies the second equality in $\lp(\full)$. For $w = (\omega_{-1}, a_{-1}, \omega) \in \mathcal{W}_\full$, we obtain
\begin{align*}
\sum_{a} z(w, a) &= \pi(x_{-1}) p(\omega|x_{-1})\\
&= \sum_{x_{-2}} \pi(x_{-2}) p(\omega_{-1}| x_{-2}) \sigma(a_{-1} | x_{-2}, \omega_{-1}) p(\omega|x_{-1})\\
&=  \sum_{ x_{-2}} z( x_{-2},\omega_{-1} , a_{-1}) p(\omega|x_{-1}) \\
&=  \sum_{ x_{-2}} \sum_{\hat{w}, \hat{a}} z( \hat{w} , \hat{a}) p(\omega|x_{-1}) \ind\{ \hat{w} = ( x_{-2},\omega_{-1}), \hat{a} = a_{-1} \}\\
&=  \sum_{\hat{w}, \hat{a}} z( \hat{w} , \hat{a})  \sum_{ x_{-2}}   p(\omega|x_{-1})  \ind\{ \hat{w} = ( x_{-2},\omega_{-1}),\hat{a} = a_{-1} \}\\
&= \sum_{\hat{w}, \hat{a}} z( \hat{w}, \hat{a}) p(w| \hat{w}, \hat{a}).
\end{align*}
Here, the second equality follows from~\eqref{eq:balance}, the second follows from the definition of $z$, and the final equality follows from the definition of $p(w|\hat{w}, \hat{a})$. Thus, we conclude that $z$ satisfies both the equalities in $\lp(\full)$. Finally, since $\sigma \in \pers(\Phi_\full) \cap \Sigma_1$, we obtain for all $x_{1} \in \pair$ and $a ,  a' \in A$, 
\begin{align*}
    \sum_{\omega} p(\omega| x_{-1}) \sigma( a | x_{-1}, \omega) {\partial u} (\omega, a, a') \geq 0.
\end{align*}
Thus, we obtain for all $x_{-1} \in \pair$ and for all $a, a' \in A$, 
\begin{align*}
    \sum_{w \in \mathcal{W}_\full} z(w, a)  D(x_{-1}, w) {\partial u} (w, a, a')    
    &= \sum_{x \in \pair} \sum_{\omega \in \Omega} z(x,\omega, a) \ind\{x = x_{-1}\}{\partial u} (\omega, a, a') \\
    &= \sum_{\omega} z(x_{-1}, \omega, a)  {\partial u} (\omega, a, a') \\
    &= \pi(x_{-1}) \left(\sum_{\omega} p(\omega| x_{-1}) \sigma( a | x_{-1}, \omega) {\partial u} (\omega, a, a')\right) \geq 0.
\end{align*}
Thus, we obtain that $z$ satisfies the inequality in $\lp(\full)$. Finally, since $\sum_{x \in \pair} z(x, \omega,a) = \pi(\omega, a)$, we conclude that the sender's payoff under $\sigma$ equals the $\lp(\full)$ objective at $z$. This concludes the first part of the statement. 

Conversely, suppose $z$ is a feasible solution for $\lp(\full)$. Define the signaling mechanism $\sigma \in \Sigma_1$ as follows: for all $w = (x_{-1}, \omega) \in \mathcal{W}_\full$ with $\sum_{a'} z(w,a') > 0$, let 
\begin{align*}
    \sigma(a|x_{-1}, \omega) \defeq \frac{z(w, a)}{\sum_{a'} z(w,a')}.
\end{align*}
For $w = (x_{-1}, \omega) \in \mathcal{W}_\full$ with  $\sum_{a'} z(w,a') = 0$, let $\sigma$ recommend the receiver-optimal action at $\omega$. We note that $\pi \in \Delta(\pair)$ defined as $\pi(\omega, a) \defeq \sum_{x \in \pair} z(x, \omega , a)$ for $(\omega, a) \in \pair$ is invariant under $\sigma$. To see this, observe
\begin{align*}
    \sum_{x_{-1}} \pi(x_{-1}) p(\omega| x_{-1}) \sigma(a| x_{-1}, \omega) 
    &=  \sum_{x_{-1}} \left( \sum_{x_{-2}}z(x_{-2}, x_{-1}) \right) p(\omega| x_{-1}) \sigma(a| x_{-1}, \omega) \\
    &=  \sum_{x_{-1}} \left(\sum_{a'} z(x_{-1}, \omega, a')\right) \sigma(a| x_{-1}, \omega) \\
    &= \sum_{x_{-1}} z(x_{-1}, \omega, a)\\
    &= \pi(\omega, a)
\end{align*}
where the first and the fourth equality follows from the definition of $\pi$. The second equality follows from the first equality constraint of $\lp(\full)$ (and from the feasibility of $z$), and the third equality follows from the definition of $\sigma$. Thus, $\pi = \steady(\sigma)$. Finally, for any $x_{-1} \in \pair$ and all $a, a' \in A$, we have
\begin{align*}
    \sum_{w \in \mathcal{W}_\full} z(w, a)  D(x_{-1}, w) {\partial u} (w, a, a')    
    &=  \sum_{\omega \in \Omega} z(x_{-1},\omega, a) {\partial u} (\omega, a, a') \\
    &= \sum_{\omega} \left( \sum_{a'} z(x_{-1}, \omega, a' )\right) \sigma(a | x_{-1}, \omega)  {\partial u} (\omega, a, a') \\
    &= \sum_{\omega} \left( \sum_{x_{-2}} z(x_{-2}, x_{-1}) p( \omega| x_{-1}) \right) \sigma(a | x_{-1}, \omega)  {\partial u} (\omega, a, a') \\
    &= \pi(x_{-1}) \cdot \sum_{\omega} p(\omega| x_{-1}) \sigma(a|x_{-1}, \omega) {\partial u}(\omega, a, a').
\end{align*}
Here, the second equality follows from the definition of $\sigma$, the third equality follows from the feasibility of $z$ to $\lp(\full)$, and in the final equality, we have used the definition of $\pi$. Thus, for all $x_{-1} \in \pair$ with $\pi(x_{-1}) > 0$, from the feasibility of $z$, we obtain  
\begin{align*}
    \sum_{\omega} p(\omega|x_{-1}) \sigma(a| x_{-1}, \omega) {\partial u}(\omega, a, a')  \geq 0,
\end{align*}
and hence a receiver, after observing $\bar{x}_{t-1} = x_{-1}$, would find it optimal to adopt action $a$ if recommended by $\sigma$. Finally, from the fact that $\sum_{x} z(x, \omega, a) = \sum_{x} z(\omega, a, x)$, we obtain that if $\pi(x_{-1}) = 0$, then $\sum_{a'} z(x_{-1}, \omega, a') = 0$, and hence, $\sigma$ recommends the receiver-optimal action at each $\omega$. Thus, again, the receiver finds it optimal to follow the recommendation. Taken together, we conclude that $\sigma \in \pers(\Phi_\full)$. The converse follows from noticing that the sender's payoff under $\sigma$ equals the $\lp(\full)$ objective at $z$. 

Summarizing the two parts, we obtain that the sender's problem $\mpp(\Phi_\full)$ can be equivalently formulated as the $\lp(\full)$.\qed
\end{proof}

\subsection{Proofs from Section~\ref{sec:analysis-partial-info}}
\label{ap:sec-partial-info}

\begin{proof}[Proof of Lemma~\ref{lem:opt}]
  We first show that for any $\ell \geq 1$, $\pers(\Phi_\full) \subseteq \pers(\Phi_\ell)$. To see this, let $\sigma \in \Sigma_k \cap \pers(\Phi_\full)$. Define $\widehat{\sigma}$ to be the
  signaling mechanism that, at each time $t$, in addition to
  recommending an action according to $\sigma$ also truthfully
  reveals $\bar{h}_t^\ell$. Since the information of the receiver under $\widehat{\sigma}$ in the model $\Phi_\ell$ is same as that under $\sigma$ in the model $\Phi_\full$, we conclude that it is optimal for the receiver to follow the recommended action. Since this is true no matter the realization of the slice $\bar{h}_t^\ell$, the receiver should find it optimal to follow the recommendation even without being informed about the realization. In other words, the receiver should find it optimal to follow the recommendations of $\sigma$ in the model $\Phi_\ell$, and hence  $\sigma \in \pers(\Phi_\ell)$. Thus, we conclude $\pers(\Phi_\full) \subseteq \pers(\Phi_\ell)$ for $\ell \geq 1$. A similar argument yields $\pers(\Phi_\ell) \subseteq \pers(\Phi_{\ell +1}) \subseteq \pers(\Phi_\no)$. 
\qed
\end{proof}

\begin{proof}[Proof of Proposition~\ref{prop:full-equals-no}]
From Lemma~\ref{lem:opt}, we have $\opt(\Phi_\full) \leq  \opt(\Phi_\no)$; thus, it remains to show that under the conditions of the lemma, $\opt(\Phi_{\full}) \geq \opt(\Phi_{\no})$. To show this inequality, we construct a signaling mechanism $\widehat{\sigma} \in \pers(\Phi_\full) \cap \Sigma_1$ that achieves the same payoff as the optimal signaling mechanism $\sigma \in \pers(\Phi_\no) \cap \Sigma_0$. 

To begin, note that since $p(\cdot|x) \in \conv(\mathcal{B}_\no)$ for each $x \in \pair$, there exists a set of non-negative weights $\{\lambda(a|x) : a \in A, x \in \pair\}$ such that 
\begin{align*}
    p(\omega|x) &= \sum_{a'}\lambda(a'|x) \mu_{a'}(\omega), \quad \text{for all $x \in \pair, \omega \in \Omega$}\\
    \sum_{a\in A}\lambda(a|x) &= 1, \quad \text{for all $x \in \pair$}.
\end{align*}
Define the mechanism $\widehat{\sigma} \in \Sigma_1$ as follows: for each $x_{-1} \in \pair$, let
 \begin{align*}
    \widehat{\sigma}(a|x_{-1}, \omega) \defeq \frac{\lambda(a|x_{-1})\mu_a(\omega)}{\sum_{a'} \lambda(a'|x_{-1})\mu_{a'}(\omega)} = \frac{\lambda(a|x_{-1})\mu_a(\omega)}{p(\omega|x_{-1})},
\end{align*}
whenever the denominator is positive, and otherwise let $\widehat{\sigma}$ recommend the receiver-optimal action at $\omega$.

We first show that $\widehat{\sigma} \in \pers(\Phi_\full)$. For each $x_{-1} \in \pair$,  we have 
\begin{align*}
    &\sum_{\omega} p(\omega|x_{-1})\widehat{\sigma}(a|x_{-1}, \omega) \partial u(\omega, a, a')\\
&\quad = \sum_{\omega} \left(\sum_{a' \in A} \lambda(a'|x_{-1}) \mu_{a'}(\omega)\right)\frac{\lambda(a|x_{-1})\mu_a(\omega)}{\sum_{a' \in A} \lambda(a'|x_{-1})\mu_{a'}(\omega)}\partial u(\omega, a, a') \\
&\quad = \lambda(a|x_{-1}) \left(\sum_{\omega}\mu_a(\omega)\partial u(\omega, a, a')\right).
\end{align*}
Since $\mu_a(\omega)$ is the posterior belief induced by $\sigma \in \pers(\Phi_\no)$, we have $\sum_{\omega}\mu_a(\omega)\partial u(\omega, a, a') \geq 0$. As $\lambda(a|x_{-1}) \geq 0$, we have for all $x_{-1} \in \pair$,
\begin{align*}
    \sum_{\omega} p(\omega|x_{-1})\widehat{\sigma}(a|x_{-1}, \omega) \partial u(\omega, a, a') \geq 0,
\end{align*}
If $\sum_{\omega'}p(\omega'|x_{-1}) \widehat{\sigma}(a| x_{-1}, \omega') > 0$, then upon dividing by it, we obtain action $a$ is optimal for the receiver if it is recommended by $\widehat{\sigma}$ in the model $\Phi_\full$. On the other hand, if $\sum_{\omega'}p(\omega'|x_{-1}) \widehat{\sigma}(a| x_{-1}, \omega') = 0$, then $\widehat{\sigma}$ recommends the receiver-optimal action. Thus, we conclude that it is always optimal for the receiver to follow the recommendations by $\widehat{\sigma}$ in the model $\Phi_\full$, and hence $\widehat{\sigma} \in \pers(\Phi_\full)$.

We now show that $\sigma$ and $\widehat{\sigma}$ induce the same marginal distribution over $\pair$. Let $\pi =\steady(\sigma)$. Let $\tau_a \defeq \sum_{\omega} \pi(\omega, a)$, and note that $\pi(\omega, a) = \tau_a \mu_a(\omega)$. Using the definition of $\widehat{\sigma}$, we have for $\omega \in \Omega$ and $a \in A$,
\begin{align}
    \sum_{x_{-1} \in \pair}\pi(x_{-1}) p(\omega|x_{-1}) \widehat{\sigma}(a|x_{-1}) &=\sum_{x_{-1} \in \pair}\pi(x_{-1}) \lambda(a|x_{-1})\mu_a(\omega)\notag \\
    &= \sum_{x_{-1} \in \pair}\tau_{a_{-1}} \mu_{a_{-1}}(\omega_{-1}) \lambda(a|x_{-1})\mu_a(\omega),\label{eq:tau-balance}
  \end{align}
  where we have used the fact that $\pi(x_{-1}) = \tau_{a_{-1}} \mu_{a_{-1}}(\omega_{-1})$. Summing both sides over $a$, we obtain 
  \begin{align*}
      \sum_{x_{-1}\in \pair}\pi(x_{-1}) p(\omega|x_{-1})  &= \sum_{x_{-1} \in \pair}\tau_{a_{-1}} \mu_{a_{-1}}(\omega_{-1}) \left(\sum_{a} \lambda(a|x_{-1})\mu_a(\omega)\right).
  \end{align*}
Moreover, from the balance equation~\eqref{eq:balance}, we have 
\begin{align*}
    \sum_{x_{-1}\in \pair}\pi(x_{-1}) p(\omega|x_{-1})  &= \sum_{a}\pi(\omega, a) =\sum_{a}\tau_a \mu_a(\omega). 
\end{align*}
Equating the right-hand sides of the two preceding equations, we obtain
\begin{align*}
     \sum_{a} \left(\tau_a - \sum_{x_{-1} \in \pair} \tau_{a_{-1}} \mu_{a_{-1}}(\omega_{-1}) \lambda(a|x_{-1}) \right) \mu_a(\omega) = 0
\end{align*}
Because $\{\mu_a\}$ are linearly independent, we have for all $a \in A$, 
\begin{align*}
    \tau_a = \sum_{x_{-1} \in \pair} \tau_{a_{-1}} \mu_{a_{-1}}(\omega_{-1}) \lambda(a|x_{-1}).
\end{align*}
Substituting back in \eqref{eq:tau-balance}, we obtain
\begin{align}
    \sum_{x_{-1} \in \pair}\pi(x_{-1}) p(\omega|x_{-1}) \widehat{\sigma}(a|x_{-1}) &= \tau_a \mu_a(\omega) = \pi(\omega, a).
  \end{align}
Thus, $\pi$ is also an invariant distribution under $\widehat{\sigma}$, and thus, the two mechanisms induce the same marginal distribution over $\Delta(\pair)$. \qed

\end{proof}

\section{Proofs from the Section~\ref{sec:robustness}}\label{ap:robustness}
In this section, we provide the missing proofs from
Section~\ref{sec:robustness}. Throughout, we use the same notation as
in that section.

\begin{proof}[Proof of Proposition~\ref{prop:partial-bilinear}]
The proof of the proposition is similar to that of Proposition~\ref{prop:lp-formulations}, and we only highlight the parts that are different. 

First, consider a signaling mechanism $\sigma \in \Sigma_1 \cap \pers(\Phi_1)$, and let $\pi = \steady(\sigma) \in \Delta(\pair^2)$ denote the invariant distribution under $\sigma$. Define $z$ as follows:
\begin{align*}
    z(x_{-2}, x_{-1}, \omega, a) \defeq \pi(x_{-2}, x_{-1}) p(\omega| x_{-1}) \sigma(a| x_{-1}, \omega).
\end{align*}
Then, similar arguments to Proposition~\ref{prop:lp-formulations} shows that $z$ satisfies the linear equalities and the inequality in \eqref{opt:non-linear}. Similarly, the value of the objective is readily seen to equal the sender's payoff under $\sigma$. Finally, the non-linear equality holds because, for any $x_{-2}', x_{-2}, x_{-1}, (\omega, a) \in \pair$, we have
\begin{align*}
       z(x_{-2}, x_{-1}, \omega, a)\cdot \sum_{a' \in A} z(x'_{-2}, x_{-1}, \omega, a') &= \pi(x_{-2}, x_{-1})p(\omega|x_{-1})\sigma(a|x_{-1}, \omega)\pi(x'_{-2}, x_{-1}) p(\omega|x_{-1})\\
       & = z(x'_{-2}, x_{-1}, \omega, a) \cdot \sum_{a' \in A} z(x_{-2}, x_{-1}, \omega, a').
\end{align*}

Conversely, let $z$ be any feasible solution to~\eqref{opt:non-linear}. We define $\sigma \in \Sigma_1$ as follows: For any $x_{-1} \in \pair$ and $\omega \in \Omega$, if there exists an $x_2 \in \pair$ such that $\sum_{a' \in A} z(x_{-2}, x_{-1}, \omega, a') > 0$, let 
\begin{align*}
\sigma(a| x_{-1}, \omega) = \frac{ z(x_{-2}, x_{-1}, \omega, a)}{\sum_{a' \in A} z(x_{-2}, x_{-1}, \omega, a')}.
\end{align*}   
Note that the non-linear constraint on $z$ implies that the right-hand side does not depend on $x_{-2}$, and thus $\sigma$ is well-defined. On the other hand, if $\sum_{a' \in A} z(x_{-2}, x_{-1}, \omega, a') = 0$ for all $x_{-2} \in \pair$, let $\sigma(\cdot| x_{-1}, \omega)$ recommend the receiver-optimal action at $\omega$. Furthermore, define $\pi \in \Delta(\pair^2)$ as $\pi(x_{-1}, \omega, a) \defeq \sum_{x_{-2}} z(x_{-2}, x_{-1}, \omega, a)$. Using similar arguments in Proposition~\ref{prop:lp-formulations}, it follows that $\pi =\steady(\sigma)$ and furthermore that $\sigma \in \pers(\Phi_1)$. The final step is to see that the objective of the non-linear program at $z$ equals the sender's payoff under $\sigma$. \qed
\end{proof}

\begin{proof}[Proof of Lemma~\ref{lem:robust-partial}] Let $\sigma \in \rp(\epsilon)$ for some fixed $\epsilon > 0$, and let $\ell \geq 0$ be such that $d_\ell(\sigma) \leq \epsilon$. Consider the information model $\Phi_\ell$, and assume the receivers follow the action recommendations.  From the perspective of a receiver at time $t$, the relevant information about the history $\bar{h}_{t-\ell}$  is the value $\bar{x}_{t-\ell-1}$, as earlier state-action pairs do not affect the subsequent transitions. If $\bar{x}_{t-\ell-1} = x \in \pair$, the distribution of $\bar{\omega}_t$ (and hence the receiver's belief) is given by $Q^\sigma(x, \omega)$. Thus, the receiver's belief lies within $d_\ell(\sigma)$ of the invariant distribution $\pi = \steady(\sigma)$. Since $\sigma \in \rp(\epsilon)$ and $d_\ell(\epsilon) \leq \epsilon$, we obtain that it is optimal for this receiver to follow the recommendation made by $\sigma$. Thus, we obtain $\sigma \in \pers(\Phi_\ell)$.  

To prove the bound in the lemma statement, we note that since $\sigma \in \rp(\epsilon) \subseteq \Sigma_0$, it corresponds to a stationary Markov policy, and hence the induced Markov chain over the states is ergodic by Assumption~\ref{as:unichain}. The result is then obtained using the following bound on the mixing time of this chain~\citep[Theorem~12.4]{levin2017markov}: 
\begin{align*}
     \ell \geq  \frac{1}{\gamma_\star}\log\left( \frac{2}{\epsilon \pi_{\min}(\sigma)}\right)  \implies d_\ell(\sigma) \leq \epsilon,
\end{align*}
where $\gamma_\star(\sigma)$ is the absolute spectral gap of the underlying Markov chain (i.e., the smallest value of $1 - |\lambda|$ over all non-unit eigenvalues $\lambda$ of the transition kernel matrix under $\sigma$), and $\pi_{\min}(\sigma) = \min_{\omega} \pi(\omega)$. Note that $\pi(\omega) > 0$ for all $\omega \in \Omega$ from Assumption~\ref{as:unichain}, and hence $\pi_{\min}(\sigma)$ is well defined.\qed\end{proof}

\begin{proof}[Proof of Lemma~\ref{lem:splitting-lemma-markov}] We
  begin by proving the first part of the lemma statement. Given a set
  $S$, beliefs $\{ \mu_s : s \in S\}$ and the weights $w_s \geq 0$ as
  in the lemma statement, define the distribution
  $\pi \in \Delta(\Omega \times A)$ as
  $\pi(\omega, a) = \sum_{s \in S} w_s \mu_s(\omega) \ind\{a_s = a\}$.
  We claim that $\pi$ is the steady state distribution under $\sigma$, when each receiver chooses the action $a_s$ after receiving signal $s \in S$. (We show below that this is indeed optimal for the receiver in the information model $\Phi_\no$.) This follows from
  noticing that $\pi$ satisfies the balance equations, as
  we show next. For each $\omega \in \Omega$, we have
  \begin{align*}
    \sum_{\omega', a'} \pi(\omega', a') p(\omega | \omega', a') 
    &= \sum_{\omega', a'}  \left(\sum_{s \in S} w_{s}\mu_{s}(\omega')\ind\{a_{s} =a'\}\right) p(\omega | \omega', a')\\
    &= \sum_{\omega'}  \sum_{s \in S} w_{s}\mu_{s}(\omega') \left(\sum_{a'} \ind\{a_{s} =a'\}p(\omega | \omega', a')\right)\\
    &= \sum_{\omega'}  \sum_{s \in S} w_{s}\mu_{s}(\omega') p(\omega | \omega', a_{s})\\
    &= \sum_{s \in S} w_{s} \mu_{s}(\omega).
  \end{align*}
  Here, the final equality follows from the assumption made on the weights
  $\{w_s\}$ in the lemma statement. This, in turn implies that
  \begin{align*}
    \sum_{\omega', a'} \pi(\omega', a') p(\omega | \omega', a') \sum_{s' \in S} \sigma(s'|\omega) \ind\{a_{s'} = a\} &= \sum_{s \in S} w_{s} \mu_{s}(\omega) \sum_{s' \in S} \sigma(s'|\omega) \ind\{a_{s'} = a\}\\
                                                                                                                     &= \sum_{s' \in S} \ind\{a_{s'} = a\} \sigma(s'|\omega) \left(\sum_{s \in S} w_{s} \mu_{s}(\omega)  \right) \\
                                                                                                                     &= \sum_{s' \in S} \ind\{a_{s'} = a\} w_{s'} \mu_{s'}(\omega) \\
    &= \pi(\omega, a),
  \end{align*}
  where the penultimate inequality follows from the definition of
  $\sigma(s|\omega)$. Thus, we conclude that $\pi$ satisfies the
  balance equations.

  Since $\pi$ is the invariant distribution under $\sigma$, the marginal distribution of the state (and hence the receivers' prior belief in the model $\Phi_\no$) equals $\pi(\omega) = \sum_{a \in A} \pi(\omega, a) = \sum_{s \in S} w_s \mu_s(\omega)$. From the definition of $\sigma$, we have
  $\pi(\omega)\sigma(s|\omega) = w_s \mu_s(\omega)$. Thus, in the information model $\Phi_\no$, the
  posterior belief of a receiver that the state is $\omega$ upon receiving a signal
  $s \in S$ is given by Bayes' rule as
  \begin{align*}
    \frac{\pi(\omega) \sigma(s|\omega)}{\sum_{\omega'} \pi(\omega') \sigma(s|\omega')} &= \mu_{s}(\omega).
  \end{align*}
  Since $\mu_s \in \mathcal{P}_{a_s}$, we conclude that choosing action $a_s$ after receiving the signal $s$ is indeed optimal for the receiver. This concludes the proof of the first part of the lemma statement.

  To show the converse, let $\sigma \in \pers(\Phi_\no) \cap \Sigma_0$ and let $\pi = \steady(\sigma)$. For any $a \in A$ with
  $\sum_{\omega'}\pi(\omega', a) >0$, define
  $w_a \defeq \sum_{\omega'}\pi(\omega', a)$, and
  $\mu_a \defeq \frac{\pi(\cdot, a)}{w_a}$. For any $a \in A$ with
  $\sum_{\omega'}\pi(\omega', a) = 0$, define $w_a = 0$ and $\mu_a$ be
  any belief in $\mathcal{P}_a$.
By construction, we have $\sum_{a\in A} w_a = 1$ and $\pi(\omega, a) = w_a \mu_a(\omega)$ for all $(\omega, a) \in \Omega \times A$. As $\sigma \in \pers(\Phi_\no) \cap \Sigma_0$, we have $\sum_{\omega} \pi(\omega, a) {\partial u}(\omega, a, a') \geq 0$ for all $a, a' \in A$. Consequently, we obtain $w_a \sum_{\omega} \mu_a(\omega){\partial u}(\omega, a, a') \geq 0$, implying that $\mu_a \in \mathcal{P}_a$ also holds for any $a$ with $w_a  > 0$. Finally, for all $\omega \in \Omega$, we obtain
\begin{align*}
 \sum_a w_a \mu_a(\omega) &= \sum_a \pi(\omega, a) \\
 &= \sum_a \sum_{\omega', a'} \pi(\omega', a')  p(\omega |\omega', a') \sigma(a |  \omega)\\
 &= \sum_{\omega', a'} \pi(\omega', a')  p(\omega |\omega', a')\\
  &=\sum_{\omega', a'} w_{a'} \mu_{a'} (\omega') p(\omega |\omega', a').
\end{align*}
Here, the second equality follows because $\pi$ is the steady-state distribution induced by $\sigma$. Thus, we conclude that $\{w_a , \mu_a\}_a$ satisfies all the conditions in the lemma statement.\qed
\end{proof}

\begin{proof}[Proof of
  Theorem~\ref{thm:robustness-history-independent}]
  To begin, let $\sigma \in \pers(\Phi_\no)$ denote the optimal
  signaling mechanism in the no-history information model $\Phi_\no$.
  Let $\pi = \steady(\sigma)$ denote the invariant distribution
  under $\sigma$, and let $\pi(\omega) = \sum_a \pi(\omega, a)$ denote
  the marginal over the states. From
  Lemma~\ref{lem:splitting-lemma-markov} we know there exist weights
  $w_a \geq 0$, with $\sum_{a \in A} w_a = 1$, and beliefs
  $\mu_a \in \mathcal{P}_a$ for $a \in A$, satisfying
  $\pi(\omega, a) = w_a \mu_a(\omega)$ and
  \begin{align}\label{eq:splitting-converse-condition}
    \sum_{\omega, a} w_a \mu_a(\omega) p(\cdot | \omega,a) = \sum_{a} w_a \mu_a.
  \end{align}

  Let $A_+ = \{ a \in A: \sum_{\omega \in \Omega} \pi(\omega) \sigma(a
  |\omega) > 0\}$ denote the set of actions that are recommended with
  positive probability under $\sigma$. It is straightforward to show
  that $A_+ = \{a : w_a > 0\}$.

  \textbf{Construction of a signaling mechanism:} We begin by
  constructing a signaling mechanism $\widehat{\sigma}$ and show it to
  be persuasive in the no-information model $\Phi_\no$; subsequently, we prove
  the stronger claim of $\epsilon$-robust persuasiveness. First, using
  Assumption~\ref{as:regularity}, for any $a \in A_+$, let
  $\eta_a \in \mathcal{P}_a$ be such that
  $\mathsf{B}_1(\eta_a, D) \subseteq \mathcal{P}_a$. For some small
  $\delta \in [0,1]$, whose exact value we will set later to obtain
  robustness, define $\xi_a = (1-\delta)\mu_a + \delta \eta_a$ for all
  $a \in A_+$. Since $\mu_a, \eta_a \in \mathcal{P}_a$ and the latter
  set is convex, we obtain that
  $\mathsf{B}_1(\xi_a, \delta D) \subseteq \mathcal{P}_a$. Next, let
  $e_\omega$ denote the belief that puts all its weight on the state
  $\omega \in \Omega$.

  We seek to construct a signaling mechanism $\widehat{\sigma}$ which
  sends signals in the set $S = A_+ \cup \Omega$, such that in the
  model $\Phi_\no$, the posterior belief upon receiving a signal
  $s = a \in A_+$ is $\xi_a$, whereas upon receiving a signal
  $s = \omega \in \Omega$, the posterior belief is $e_\omega$. Let
  $a_s = a$ if $s = a \in A_+$ and $a_s = a_\omega$ for
  $s = \omega \in \Omega$, where $a_\omega$ denotes an optimal action
  for the receiver at state $\omega$. Using
  Lemma~\ref{lem:splitting-lemma-markov}, there exists a signaling
  mechanism $\widehat{\sigma}$ inducing the aforementioned beliefs in
  steady state if there exist weights $\{ \hat{w}_s : s \in S\}$ with
  $\sum_{s \in S} \hat{w}_s = 1$, such that
  \begin{align}\label{eq:splitting-condition}
    \sum_{a \in A_+} \sum_{\omega} \hat{w}_a \xi_a(\omega) p(\cdot | \omega, a) + \sum_{\omega} \hat{w}_{\omega} p(\cdot | \omega, a_\omega)  = \sum_{a \in A_+} \hat{w}_a \xi_a + \sum_{\omega} \hat{w}_{\omega} e_\omega.
  \end{align}
  To produce such weights, we first define $\{\hat{w}_a \}$ in terms
  of $\{\hat{w}_\omega \}$. Let
  $\hat{w}_a = (1 - \sum_{\omega} \hat{w}_\omega)w_a$ for each
  $a \in A_+$. Since $\sum_a w_a = 1$, it follows that the weights
  $\{ \hat{w}_s\}$ sum to one as well. Further, to simplify
  expressions, let $\rho \defeq \sum_\omega \hat{w}_\omega$. Then,
  after moving all terms containing $\hat{w}_\omega$ on one side, the
  condition~\eqref{eq:splitting-condition} becomes
  \begin{align}\label{eq:splitting-condition-simple} \frac{1}{1-\rho}
    \left(\sum_{\omega} \hat{w}_{\omega} e_\omega - \sum_{\omega}
    \hat{w}_{\omega} p(\cdot | \omega, a_\omega) \right)
    &= \sum_{a \in A_+} \sum_{\omega} w_a \xi_a(\omega) p(\cdot | \omega, a) - \sum_{a \in A_+} w_a \xi_a \notag\\
    &= \delta \left(\sum_{a \in A_+} \sum_{\omega} w_a \eta_a(\omega)
      p(\cdot | \omega, a) - \sum_{a \in A_+} w_a \eta_a\right),
  \end{align}
  where, in the second equality, we have used
  $\xi_a = (1-\delta) \mu_a + \delta \eta_a$, along with the fact that
  $\{ \mu_a, w_a\}_a$ satisfy \eqref{eq:splitting-converse-condition}.

  In Lemma~\ref{lem:bound-LP}, we show that there exists
  $y = (y_\omega \geq 0 : \omega \in \Omega)$ satisfying
  \begin{align}\label{eq:lp-constraint}
    \sum_{\omega} y_{\omega} e_\omega -  \sum_{\omega} y_{\omega} p(\cdot | \omega, a_\omega) 
    &= \sum_{a \in A_+} \sum_{\omega} w_a \eta_a(\omega) p(\cdot | \omega, a) - \sum_{a \in A_+} w_a \eta_a.
  \end{align}
  For any such $y$, we obtain that $\hat{w}_\omega = \frac{\delta y_\omega}{1 + \delta \|y\|_1}$ and
  $\rho = \sum_{\omega} \hat{w}_\omega = \frac{\delta \|y\|_1}{1 +
    \delta \|y\|_1}$ form a solution to
  \eqref{eq:splitting-condition-simple}, and hence, there exist
  weights satisfying~\eqref{eq:splitting-condition}.

  Thus, by Lemma~\ref{lem:splitting-lemma-markov}, we obtain the
  existence of a history-independent signaling mechanism $\widehat{\sigma}$ sending
  signals $s \in S = A_+ \cup \Omega$, such that in the no-history information model $\Phi_\no$, the posterior beliefs lie in the set
  $\{ \xi_a : a \in A_+\} \cup \{ e_\omega : \omega \in \Omega\}$. The
  mechanism $\widehat{\sigma}$ sends signals with the following
  probabilities: for each $\omega \in \Omega$:
  \begin{align}\label{eq:constructed-sigma}
  \widehat{\sigma}(s | \omega) \defeq 
  \begin{cases}
    \frac{\hat{w}_a \xi_a(\omega)}{\sum_{a'} \hat{w}_{a'} \xi_{a'}(\omega) + \hat{w}_\omega} & \text{for $s \in A_+$;}\\
    \frac{\hat{w}_\omega}{\sum_{a'} \hat{w}_{a'} \xi_{a'}(\omega) + \hat{w}_\omega} & \text{if $s = \omega$;}\\
    0, & \text{otherwise.}
  \end{cases}
  \end{align}
  Here, we interpret the signal $s = a \in A_+$ as an direct
  recommendation to choose action $a$. On the other hand, the signal
  $s = \omega \in \Omega$ fully reveals the state and is interpreted
  as a recommendation to choose the action $a_\omega$. Since
  $\xi_a \in \mathcal{P}_a$ for each $a \in A_+$ and $a_\omega$ is an
  optimal action for the receiver at state $\omega \in \Omega$, we
  conclude that the mechanism $\widehat{\sigma}$ is persuasive, in the
  sense that, the receiver finds it optimal to follow the
  recommendation.

  Note that Lemma~\ref{lem:splitting-lemma-markov} implies that the
  invariant distribution $\widehat{\pi}$ under $\widehat{\sigma}$
  is given by
  $\widehat{\pi}(\omega,a) = \hat{w}_a \xi_a(\omega)\ind\{a \in A_+\}
  + \hat{w}_\omega \ind\{a = a_\omega\}$ for each $\omega \in \Omega$
  and $a \in A$. Let $\widehat{\pi}(\omega) \defeq \sum_{a \in A} \widehat{\pi}(\omega,a) =
  \sum_{a \in A_+} \hat{w}_a \xi_a(\omega) + \hat{w}_\omega$.

  \textbf{Robustness:} Next, we show that $\widehat{\sigma} \in \rp(\epsilon)$. Suppose a receiver's belief about $\bar{\omega}_{t}$ is given by a distribution
  $\pi' \in \Delta(\Omega)$, with $\| \pi' - \widehat{\pi}\|_1 \leq \epsilon$. Upon receiving a signal $s = \omega \in S$ from the signaling
  mechanism $\widehat{\sigma}$, it is straightforward to see that the
  receiver's belief about $\bar{\omega}_t$ continues to update to
  $e_\omega$, and hence $a_{\omega}$ is still optimal for the receiver
  on receiving signal $s = \omega$.

  On the other hand, upon receiving a signal $s = a \in A_+$, the
  receiver's belief about $\bar{\omega}_t$ updates to $\xi'_a$,
  obtained via Bayes' rule as
  $\xi'_a(\omega) = \frac{ \mu'(\omega)
    \widehat{\sigma}(a|\omega)}{\sum_{\omega'}\mu'(\omega')
    \widehat{\sigma}(a|\omega')}$. Using a similar argument as in
  \citep{zu2021learning}, we obtain that the following bound on the
  $\ell_1$ distance between $\xi'_a$ and $\xi_a$ (for the sake of
  completeness, we include the algebraic steps in
  Lemma~\ref{lem:bayes-continuity}):
  \begin{align*}
    \|\xi'_a - \xi_a \|_1 &\leq 2 \left(\sup_{\omega \in \Omega} \frac{\xi_a(\omega)}{\widehat{\pi}(\omega)} \right)  \cdot \| \pi' - \widehat{\pi}\|_1. 
\end{align*}
Since
$\widehat{\pi}(\omega) = \sum_{a \in A} \widehat{\pi}(\omega, a) =
\sum_{a \in A_+} \hat{w}_a \xi_a(\omega) + \hat{w}_\omega$ for each
$\omega \in \Omega$, we have for each $a \in A_+$,
\begin{align*}
  \sup_{\omega \in \Omega} \frac{\xi_a(\omega)}{\widehat{\pi}(\omega)} &= \sup_{\omega \in \Omega} \frac{\xi_a(\omega)}{\sum_{a' \in A_+} \hat{w}_{a'} \xi_{a'}(\omega) + \hat{w}_\omega} \leq  \frac{1}{\hat{w}_a} \leq \frac{1}{(1- \rho)w_{\min}},
\end{align*}
where we use the fact that $\hat{w}_a = (1 - \rho) w_a$ for
$a \in A_+$, and define $w_{\min} \defeq \min_{a \in A_+} w_a$. Thus,
we obtain
\begin{align*}
  \|\xi_a' - \xi_a\|_1
  &\leq \frac{2 \|\pi' - \widehat{\pi}\|_1}{(1- \rho)w_{\min}}
    \leq \frac{2 (1 + \delta \|y\|_1)  \epsilon }{w_{\min} },
\end{align*}
where we have substituted
$\rho = \frac{\delta \|y\|_1}{1 + \delta \|y\|_1}$, and used
$\|\pi' - \widehat{\pi}\|_1 \leq \epsilon$.


For $\epsilon \leq \frac{w_{\min}D}{2(1 + \|y\|_1)}$, choosing
$\delta = \frac{2 \epsilon }{w_{\min}D - 2 \epsilon \|y\|_1} \in
[0,1]$, we obtain that $\|\xi'_a - \xi_a\|_1 \leq \delta D$ and hence
$\xi_a' \in \mathsf{B}_1(\eta_a, \delta D) \subseteq \mathcal{P}_a$.
Hence, starting with a prior $\pi'$ with
$\| \pi' - \widehat{\pi}\| \leq \epsilon$, the posterior belief upon
receiving a signal $s = a \in A_+$ lies in the set $\mathcal{P}_a$,
implying that the action $a$ continues to be optimal for the receiver.
Taken together, the signaling mechanism $\widehat{\sigma}$ is persuasive
for all beliefs $\pi' \in \mathsf{B}_1(\widehat{\pi}, \epsilon)$, and
hence is $\epsilon$-robustly persuasive.

\textbf{Bound on sender's payoff:} Finally, we provide a bound on the
sender's expected utility under the signaling mechanism
$\widehat{\sigma}$, as follows:
\begin{align*}
  \sum_{\omega \in \Omega} \sum_{a \in A} \widehat{\pi}(\omega, a) v(\omega, a)
  &= \sum_{\omega} \sum_{a \in A_+} \hat{w}_a \xi_a(\omega)v(\omega, a)  + \sum_{\omega}   \hat{w}_\omega v(\omega, a_\omega)\\
  &\geq \sum_{\omega} \sum_{a \in A_+} \hat{w}_a \xi_a(\omega)v(\omega, a)  \\
  &\geq (1 - \rho) (1- \delta) \sum_{\omega} \sum_{a \in A_+}   w_a \mu_a(\omega)v(\omega, a) \\
  &= (1 - \rho) (1- \delta) \opt(\Phi_\no),
\end{align*}
where in the second inequality, we use $\hat{w}_a = (1 - \rho) w_a$
and $\xi_a \geq (1- \delta) \mu_a$. 

Substituting for $\delta$ and $\rho$, we obtain the sender's payoff is
lower-bounded by
\begin{align*}
  \frac{1- \delta}{1+ \delta \|y\|_1} \opt(\Phi_\no) =   \left(1 - \frac{2  (1 + \|y\|_1)  }{w_{\min}D} \epsilon\right) \opt(\Phi_\no).
\end{align*}

In Lemma~\ref{lem:bound-LP}, we show that there exists a solution
$y \geq 0$ to \eqref{eq:lp-constraint} satisfying
$\|y \|_1 \leq \frac{2(1+ \tau)\sqrt{|\Omega|}}{s_f}$. Thus, we obtain
that the sender's expected payoff under $\widehat{\sigma}$ is
lower-bounded by
\begin{align*}
  \left(1 - \frac{2 \epsilon}{w_{\min}D} \left(1 + \frac{2(1+ \tau)\sqrt{|\Omega|}}{s_f}    \right) \right) \cdot \opt(\Phi_\no).
\end{align*}
This completes the proof.\qed
\end{proof}

The following lemma is used in the proof of
Theorem~\ref{thm:robustness-history-independent}.
\begin{lemma}\label{lem:bound-LP}
  Consider the following linear program:
  \begin{align}\label{small-perturbation-LP}
    \min_{y \geq 0} &\qquad  \sum_{\omega \in \Omega} y_\omega  \notag\\
    \sum_{\omega} y_{\omega} e_\omega -  \sum_{\omega} y_{\omega} p(\cdot | \omega, a_\omega) 
                    &= \sum_{a \in A_+} \sum_{\omega} w_a \eta_a(\omega) p(\cdot | \omega, a) - \sum_{a \in A_+} w_a \eta_a.
  \end{align}
  The preceding linear program is feasible, and its optimal solution
  is upper bounded by $\frac{2(1+ \tau)\sqrt{|\Omega|}}{s_f}$, where
  $\tau = \max_\omega 1/\nu_f(\omega)$, and $s_f$ is the smallest positive
  singular value of $I - P_f$.
\end{lemma}
\begin{proof}
  To show the feasibility of he linear program~\eqref{small-perturbation-LP}, we first cast it into a matrix form. Let $P_a \in \reals^{|\Omega| \times |\Omega|}$ be the matrix
  with $P_a(\omega, \omega') \defeq p(\omega'|\omega, a)$.
  Then, the \lp~\eqref{small-perturbation-LP} can be recasted as
\begin{align*}
  \min_{y \in \reals^{|\Omega|}} \quad &  \one^T \cdot y \\
  y^T (I - P_f) &= \sum_{a \in A_+} w_a \eta_a^T (P_a - I) \\
   y &\geq 0.
\end{align*}
where $y $, and $\one \in \reals^{\Omega}$ is the all-one vector. 

Since $P_f$ and $\{P_a\}_{a \in A}$ are transition kernels, we have that $(I-P_f) \one = (I-P_a) \one = 0$ for all $a \in A$. This implies that $\one$ does not lie in the row span of $I-P_f$ and $I-P_a$ for any $a \in A_+$, and hence $\one$ is orthogonal to vector $\sum_{a \in A_+} w_a \eta_a^T (P_a - I)$. Since $P_f$ is ergodic (from Assumption~\ref{as:unichain}), we have $\mathsf{rank}(I-P_f) = |\Omega|-1$ and thus the vector $\sum_{a \in A_+} w_a \eta_a^T (P_a - I)$ lies in the row span of $I-P_f$. Therefore, the equation $y^T (I - P_f) = \sum_{a \in A_+} w_a \eta_a^T (P_a - I)$ has a feasible solution. Thus, it remains to be shown that there exists one that is non-negative.

Let $u$ be any solution to the equality in the linear program. Because $P_f$ is ergodic, there is a unique stationary distribution $\nu_f$ such that $\nu_f^T(I-P_f) = 0$ and each element $\nu_f(\omega) > 0$ for all $\omega \in \Omega$. Let $y = u + k\nu_f$ where $k = \max_{\omega:u(\omega) <0} \frac{|u(\omega)|}{\nu_f(\omega)}$ is chosen so that $y \geq 0$. Hence, we obtain that $y$ is feasible for the LP~\eqref{small-perturbation-LP}.

The proof is complete upon showing that $\|y\|_1 = \|u + k\nu_f\|_1 \leq \frac{2(1+\tau)\sqrt{|\Omega|}}{s_f}$.  To see this, note that since $\mathsf{rank}(P_f -I) = |\Omega|-1$, the matrix
 $I-P_f$ has a singular value decomposition $Q\Lambda \Tilde{Q}^T$, where $\Lambda \in \reals^{(|\Omega|-1) \times (|\Omega|-1)}$ is the diagonal matrix of singular values and $Q, \tilde{Q} \in \reals^{|\Omega| \times (|\Omega|-1)}$ are composed of orthogonal column vectors. Because multiplying by an orthogonal matrix preserves the $\ell_2$ norm, we obtain 
\begin{align*}
    \|u^T(I-P_f) \|_2  = \|u^T  Q\Lambda \Tilde{Q}^T \|_2 = \|u^T Q\Lambda \|_2
    \geq s_f \|u^T Q \|_2 =  s_f \|u\|_2 \geq \frac{s_f}{\sqrt{|\Omega|}} \|u\|_1, 
\end{align*}
where $s_f$ is the smallest diagonal element (i.e., the smallest singular value) of $\Lambda$, and the final inequality follows from the relationship between $\ell_1$ and $\ell_2$ norms. Note that from Assumption~\ref{as:unichain}, we have $s_f > 0$.
 On the other hand, 
 \begin{align*}
     \Big\lVert \sum_{a \in A_+} w_a \eta_a^T (P_a - I) \Big \rVert_1 
     &\leq \sum_{a \in A_+} w_a \| \eta_a^T (P_a - I)\|_1 \\
     &\leq \sum_{a \in A_+} w_a \sum_{j = 1}^{|\Omega|} \Big \lvert \eta_a^{(j)}(p_{jj}-1+\sum_{i \neq j} p_{ij})\Big \rvert \\
     & \leq \sum_{a \in A_+} w_a \sum_{j = 1}^{|\Omega|}  \eta_a^{(j)} (1-p_{jj} + \sum_{i \neq j} p_{ij}) \\
     &\leq 2\sum_{a \in A_+} w_a \sum_{j = 1}^{|\Omega|} \eta_a^{(i)} \\
     & =2,
 \end{align*}
where the second inequality follows triangle inequality.   Taken together, $\|u\|_1 \leq \frac{2\sqrt{|\Omega|}}{s_f}$. Finally observe that $k = \max_{\omega:u(\omega) <0} \frac{|u(\omega)|}{\nu_f(\omega)} \leq \max_{\omega} \frac{1}{\nu_f(\omega)}  \cdot  \max_{\omega:u(\omega) <0} |u(\omega)| \leq \tau \|u\|_1$. Hence, using the fact that $\|\nu_f\|_1 =1$, we obtain  $ \|y \|_1 = \| u + k \nu_f \|_1 \leq \|u \|_1 + k \|\nu_f\|_1 = \|u\|_1 + k  \leq (1+\tau) \|u\|_1 \leq \frac{2(1+  \tau)\sqrt{|\Omega|}}{s_f}$.\qed
\end{proof}

The following lemma establishes the continuity of the Bayes' update.
The proof is from \citep{zu2021learning}; we include it
here for completeness. We use the same notation as in the proof of
Theorem~\ref{thm:robustness-history-independent}.
\begin{lemma}\label{lem:bayes-continuity} For each $a \in A_+$, we
  have
  \begin{align*}
    \|\xi_a' -\xi_a\|_1 \leq 2 \left(\sup_{\omega \in \Omega} \frac{\xi_a(\omega)}{\widehat{\pi}(\omega)} \right)  \cdot \| \pi' - \widehat{\pi}\|_1.
  \end{align*}
\end{lemma}
\begin{proof} We obtain
  \begin{align*}
  \|\xi'_a - \xi_a \|_1
  &= \sum_{\omega \in \Omega} |\xi'_a(\omega) - \xi_a(\omega)| \\
  &= \sum_{\omega \in \Omega} \left| \frac{ \pi'(\omega) \widehat{\sigma}(a|\omega)}{\sum_{\omega'}\pi'(\omega') \widehat{\sigma}(a|\omega')} -   \frac{ \widehat{\pi}(\omega) \widehat{\sigma}(a|\omega)}{\sum_{\omega'}\widehat{\pi}(\omega') \widehat{\sigma}(a|\omega')} \right| \\
  &\leq \sum_{\omega \in \Omega} \left| \frac{ \pi'(\omega) \widehat{\sigma}(a|\omega)}{\sum_{\omega'}\pi'(\omega') \widehat{\sigma}(a|\omega')} -   \frac{ \pi'(\omega) \widehat{\sigma}(a|\omega)}{\sum_{\omega'}\widehat{\pi}(\omega') \widehat{\sigma}(a|\omega')} \right| \\
  &\quad + \sum_{\omega \in \Omega} \left| \frac{ \pi'(\omega) \widehat{\sigma}(a|\omega)}{\sum_{\omega'}\widehat{\pi}(\omega') \widehat{\sigma}(a|\omega')} -   \frac{ \widehat{\pi}(\omega) \widehat{\sigma}(a|\omega)}{\sum_{\omega'}\widehat{\pi}(\omega') \widehat{\sigma}(a|\omega')} \right| \\
  &\leq  \left| \sum_{\omega} \frac{\widehat{\sigma}(a|\omega)}{\sum_{\omega'}\widehat{\pi}(\omega') \widehat{\sigma}(a|\omega')} \left(\widehat{\pi}(\omega)  - \pi'(\omega)\right) \right| \\
  &\quad + \sum_{\omega \in \Omega} \frac{\widehat{\sigma}(a|\omega)}{\sum_{\omega'}\widehat{\pi}(\omega') \widehat{\sigma}(a|\omega')} \left| \pi'(\omega)  -   \widehat{\pi}(\omega) \right| \\
  &\leq 2 \sum_{\omega \in \Omega} \frac{\widehat{\sigma}(a|\omega)}{\sum_{\omega'}\widehat{\pi}(\omega') \widehat{\sigma}(a|\omega')} \left| \pi'(\omega)  -   \widehat{\pi}(\omega) \right| \\
  &\leq 2 \left(\sup_{\omega \in \Omega} \frac{\widehat{\sigma}(a|\omega)}{\sum_{\omega'}\widehat{\pi}(\omega') \widehat{\sigma}(a|\omega')} \right)  \cdot \| \pi' - \widehat{\pi}\|_1\\
   &= 2 \left(\sup_{\omega \in \Omega} \frac{\xi_a(\omega)}{\widehat{\pi}(\omega)} \right)  \cdot \| \pi' - \widehat{\pi}\|_1,
   \end{align*}
where the last equality follows from the definition of
$\xi_a(\omega)$.\qed
\end{proof}


%% file: markovian-persuasion.bbl
\begin{thebibliography}{33}
\providecommand{\natexlab}[1]{#1}
\providecommand{\url}[1]{\texttt{#1}}
\expandafter\ifx\csname urlstyle\endcsname\relax
  \providecommand{\doi}[1]{doi: #1}\else
  \providecommand{\doi}{doi: \begingroup \urlstyle{rm}\Url}\fi

\bibitem[Alizamir et~al.(2020)Alizamir, de~V{\'e}ricourt, and
  Wang]{alizamir2020warning}
Saed Alizamir, Francis de~V{\'e}ricourt, and Shouqiang Wang.
\newblock Warning against recurring risks: An information design approach.
\newblock \emph{Management Science}, 66\penalty0 (10):\penalty0 4612--4629,
  2020.

\bibitem[Anunrojwong et~al.(2022)Anunrojwong, Iyer, and
  Manshadi]{anunrojwong2022information}
Jerry Anunrojwong, Krishnamurthy Iyer, and Vahideh Manshadi.
\newblock Information design for congested social services: Optimal need-based
  persuasion.
\newblock \emph{Management Science}, 2022.

\bibitem[Ashkenazi-Golan et~al.(2022)Ashkenazi-Golan, Hern{\'a}ndez, Neeman,
  and Solan]{ashkenazi2022markovian}
Galit Ashkenazi-Golan, Pen{\'e}lope Hern{\'a}ndez, Zvika Neeman, and Eilon
  Solan.
\newblock Markovian persuasion with two states.
\newblock \emph{arXiv preprint arXiv:2209.06536}, 2022.

\bibitem[Aumann et~al.(1995)Aumann, Maschler, and Stearns]{aumannM95}
Robert~J Aumann, Michael Maschler, and Richard~E Stearns.
\newblock \emph{Repeated games with incomplete information}.
\newblock MIT press, 1995.

\bibitem[Babichenko et~al.(2022)Babichenko, Talgam-Cohen, Xu, and
  Zabarnyi]{babichenko2022regret}
Yakov Babichenko, Inbal Talgam-Cohen, Haifeng Xu, and Konstantin Zabarnyi.
\newblock Regret-minimizing bayesian persuasion.
\newblock \emph{Games and Economic Behavior}, 136:\penalty0 226--248, 2022.

\bibitem[Bergemann and Morris(2016)]{bergemann2016bayes}
Dirk Bergemann and Stephen Morris.
\newblock Bayes correlated equilibrium and the comparison of information
  structures in games.
\newblock \emph{Theoretical Economics}, 11\penalty0 (2):\penalty0 487--522,
  2016.

\bibitem[Bergemann and Morris(2019)]{bergemann2019information}
Dirk Bergemann and Stephen Morris.
\newblock Information design: A unified perspective.
\newblock \emph{Journal of Economic Literature}, 57\penalty0 (1):\penalty0
  44--95, 2019.

\bibitem[Bernasconi et~al.(2022)Bernasconi, Castiglioni, Marchesi, Gatti,
  Trov{\`o}, et~al.]{bernasconi2022sequential}
Martino Bernasconi, Matteo Castiglioni, Alberto Marchesi, Nicola Gatti,
  Francesco Trov{\`o}, et~al.
\newblock Sequential information design: Learning to persuade in the dark.
\newblock In \emph{Thirty-sixth Conference on Neural Information Processing
  Systems}, pages 1--25, 2022.

\bibitem[Bizzotto et~al.(2021)Bizzotto, R{\"u}diger, and
  Vigier]{bizzotto2021dynamic}
Jacopo Bizzotto, Jesper R{\"u}diger, and Adrien Vigier.
\newblock Dynamic persuasion with outside information.
\newblock \emph{American Economic Journal: Microeconomics}, 13\penalty0
  (1):\penalty0 179--94, 2021.

\bibitem[Board and Lu(2018)]{board2018competitive}
Simon Board and Jay Lu.
\newblock Competitive information disclosure in search markets.
\newblock \emph{Journal of Political Economy}, 126\penalty0 (5):\penalty0
  1965--2010, 2018.

\bibitem[Dughmi(2017)]{dughmi2017algorithmic}
Shaddin Dughmi.
\newblock Algorithmic information structure design: a survey.
\newblock \emph{ACM SIGecom Exchanges}, 15\penalty0 (2):\penalty0 2--24, 2017.

\bibitem[Dworczak and Pavan(2022)]{dworczak2022preparing}
Piotr Dworczak and Alessandro Pavan.
\newblock Preparing for the worst but hoping for the best: Robust (bayesian)
  persuasion.
\newblock \emph{Econometrica}, 90\penalty0 (5):\penalty0 2017--2051, 2022.

\bibitem[Ely(2017)]{ely2017beeps}
Jeffrey~C Ely.
\newblock Beeps.
\newblock \emph{American Economic Review}, 107\penalty0 (1):\penalty0 31--53,
  2017.

\bibitem[Farhadi and Teneketzis(2022)]{farhadi2022dynamic}
Farzaneh Farhadi and Demosthenis Teneketzis.
\newblock Dynamic information design: A simple problem on optimal sequential
  information disclosure.
\newblock \emph{Dynamic Games and Applications}, 12\penalty0 (2):\penalty0
  443--484, 2022.

\bibitem[Feinberg and Shwartz(1994)]{feinberg1994markov}
Eugene~A Feinberg and Adam Shwartz.
\newblock Markov decision models with weighted discounted criteria.
\newblock \emph{Mathematics of Operations Research}, 19\penalty0 (1):\penalty0
  152--168, 1994.

\bibitem[Feinberg and Shwartz(1995)]{feinberg1995constrained}
Eugene~A Feinberg and Adam Shwartz.
\newblock Constrained markov decision models with weighted discounted rewards.
\newblock \emph{Mathematics of Operations Research}, 20\penalty0 (2):\penalty0
  302--320, 1995.

\bibitem[Gan et~al.(2022)Gan, Majumdar, Radanovic, and Singla]{gan2022bayesian}
Jiarui Gan, Rupak Majumdar, Goran Radanovic, and Adish Singla.
\newblock Bayesian persuasion in sequential decision-making.
\newblock In \emph{Proceedings of the AAAI Conference on Artificial
  Intelligence}, volume 36(5), pages 5025--5033, 2022.

\bibitem[Hu and Weng(2021)]{hu2021robust}
Ju~Hu and Xi~Weng.
\newblock Robust persuasion of a privately informed receiver.
\newblock \emph{Economic Theory}, 72\penalty0 (3):\penalty0 909--953, 2021.

\bibitem[Kamenica and Gentzkow(2011)]{kamenica2011bayesian}
Emir Kamenica and Matthew Gentzkow.
\newblock Bayesian persuasion.
\newblock \emph{American Economic Review}, 101\penalty0 (6):\penalty0
  2590--2615, 2011.

\bibitem[Kosterina(2022)]{kosterina2022persuasion}
Svetlana Kosterina.
\newblock Persuasion with unknown beliefs.
\newblock \emph{Theoretical Economics}, 17\penalty0 (3):\penalty0 1075--1107,
  2022.

\bibitem[Lehrer and Shaiderman(2022)]{lehrer2022markovian}
Ehud Lehrer and Dimitry Shaiderman.
\newblock Markovian persuasion with stochastic revelations.
\newblock \emph{arXiv preprint arXiv:2204.08659}, 2022.

\bibitem[Levin and Peres(2017)]{levin2017markov}
David~A Levin and Yuval Peres.
\newblock \emph{Markov chains and mixing times}, volume 107.
\newblock American Mathematical Soc., 2017.

\bibitem[Li and Norman(2021)]{li2021sequential}
Fei Li and Peter Norman.
\newblock Sequential persuasion.
\newblock \emph{Theoretical Economics}, 16\penalty0 (2):\penalty0 639--675,
  2021.

\bibitem[Lingenbrink and Iyer(2019)]{lingenbrink2019optimal}
David Lingenbrink and Krishnamurthy Iyer.
\newblock Optimal signaling mechanisms in unobservable queues.
\newblock \emph{Operations research}, 67\penalty0 (5):\penalty0 1397--1416,
  2019.

\bibitem[Orlov et~al.(2020)Orlov, Skrzypacz, and Zryumov]{orlov2020persuading}
Dmitry Orlov, Andrzej Skrzypacz, and Pavel Zryumov.
\newblock Persuading the principal to wait.
\newblock \emph{Journal of Political Economy}, 128\penalty0 (7):\penalty0
  2542--2578, 2020.

\bibitem[Puterman(2014)]{puterman2014markov}
Martin~L Puterman.
\newblock \emph{Markov decision processes: discrete stochastic dynamic
  programming}.
\newblock John Wiley \& Sons, 2014.

\bibitem[Renault et~al.(2017)Renault, Solan, and Vieille]{renault2017optimal}
J{\'e}r{\^o}me Renault, Eilon Solan, and Nicolas Vieille.
\newblock Optimal dynamic information provision.
\newblock \emph{Games and Economic Behavior}, 104:\penalty0 329--349, 2017.

\bibitem[Tsitsiklis(2007)]{tsitsiklis}
John~N Tsitsiklis.
\newblock Np-hardness of checking the unichain condition in average cost mdps.
\newblock \emph{Operations research letters}, 35\penalty0 (3):\penalty0
  319--323, 2007.

\bibitem[Ui(2022)]{ui2022optimal}
Takashi Ui.
\newblock Optimal and robust disclosure of public information.
\newblock \emph{arXiv preprint arXiv:2203.16809}, 2022.

\bibitem[Wolff(1982)]{wolff1982poisson}
Ronald~W. Wolff.
\newblock Poisson arrivals see time averages.
\newblock \emph{Operations Research}, 30\penalty0 (2):\penalty0 223--231, 1982.
\newblock ISSN 0030364X, 15265463.
\newblock URL \url{http://www.jstor.org/stable/170165}.

\bibitem[Wu et~al.(2022)Wu, Zhang, Feng, Wang, Yang, Jordan, and
  Xu]{wu2022sequential}
Jibang Wu, Zixuan Zhang, Zhe Feng, Zhaoran Wang, Zhuoran Yang, Michael~I
  Jordan, and Haifeng Xu.
\newblock Sequential information design: Markov persuasion process and its
  efficient reinforcement learning.
\newblock In \emph{Proceedings of the 23rd ACM Conference on Economics and
  Computation}, pages 471--472, 2022.

\bibitem[Wu(2021)]{wu2021sequential}
Wenhao Wu.
\newblock Sequential bayesian persuasion.
\newblock 2021.

\bibitem[Zu et~al.(2021)Zu, Iyer, and Xu]{zu2021learning}
You Zu, Krishnamurthy Iyer, and Haifeng Xu.
\newblock Learning to persuade on the fly: Robustness against ignorance.
\newblock In \emph{Proceedings of the 22nd ACM Conference on Economics and
  Computation}, EC '21, page 927–928, New York, NY, USA, 2021. Association
  for Computing Machinery.
\newblock ISBN 9781450385541.
\newblock \doi{10.1145/3465456.3467593}.
\newblock URL \url{https://doi.org/10.1145/3465456.3467593}.

\end{thebibliography}
